\documentclass[letterpaper,english,notitlepage]{article}
\usepackage{lmodern}

\usepackage[T1]{fontenc}
\usepackage[latin9]{inputenc}
\usepackage{color}
\usepackage{babel}
\usepackage{refstyle}
\usepackage{amsthm}
\usepackage{amsmath}
\usepackage{amssymb}
\usepackage{esint}
\usepackage[unicode=true,
 bookmarks=true,bookmarksnumbered=false,bookmarksopen=false,
 breaklinks=false,pdfborder={0 0 0},backref=page,colorlinks=true]
 {hyperref}
\hypersetup{pdftitle={An Estimate on the Number of Eigenvalues of a Quasiperiodic Jacobi Matrix},
 pdfauthor={Ilia Binder and Mircea Voda}}

\makeatletter

\AtBeginDocument{\providecommand\lemref[1]{\ref{lem:#1}}}
\AtBeginDocument{\providecommand\eqref[1]{\ref{eq:#1}}}
\AtBeginDocument{\providecommand\thmref[1]{\ref{thm:#1}}}
\AtBeginDocument{\providecommand\propref[1]{\ref{prop:#1}}}
\AtBeginDocument{\providecommand\corref[1]{\ref{cor:#1}}}
\pdfpageheight\paperheight
\pdfpagewidth\paperwidth

\RS@ifundefined{subref}
  {\def\RSsubtxt{section~}\newref{sub}{name = \RSsubtxt}}
  {}
\RS@ifundefined{thmref}
  {\def\RSthmtxt{theorem~}\newref{thm}{name = \RSthmtxt}}
  {}
\RS@ifundefined{lemref}
  {\def\RSlemtxt{lemma~}\newref{lem}{name = \RSlemtxt}}
  {}

\numberwithin{equation}{section}
\numberwithin{figure}{section}
\theoremstyle{plain}
\newtheorem{thm}{\protect\theoremname}
\theoremstyle{plain}
\newtheorem{prop}[thm]{\protect\propositionname}
\theoremstyle{plain}
\newtheorem{lem}[thm]{\protect\lemmaname}
\ifx\proof\undefined\
\newenvironment{proof}[1][\protect\proofname]{\par
\normalfont\topsep6\p@\@plus6\p@\relax
\trivlist
\itemindent\parindent
\item[\hskip\labelsep
\scshape
#1]\ignorespaces
}{%
\endtrivlist\@endpefalse
}
\fi
\theoremstyle{plain}
\newtheorem{cor}[thm]{\protect\corollaryname}

\usepackage{amsfonts}
\DeclareMathOperator{\mes}{mes}
\newcommand{\CC}{\mathbb{C}}
\newcommand{\TT}{\mathbb{T}}
\newcommand{\cG}{\mathcal{G}}
\newcommand{\f}{f^{u}}
\newcommand{\fg}{M^{u}}
\newcommand{\g}{f^{a}}
\newcommand{\gh}{M^{a}}
\newcommand{\cA}{\mathcal{A}}
\newcommand{\cB}{\mathcal{B}}
\newcommand{\cM}{\mathcal{M}}
\newcommand{\cE}{\mathcal{E}}
\newcommand{\cD}{\mathcal{D}}
\numberwithin{equation}{section}
\numberwithin{thm}{section}
\makeatother

\providecommand{\corollaryname}{Corollary}
\providecommand{\lemmaname}{Lemma}
\providecommand{\proofname}{Proof}
\providecommand{\propositionname}{Proposition}
\providecommand{\theoremname}{Theorem}

\begin{document}
\let\corref=\relax
\let\propref=\relax
\newref{eq}{refcmd = {(\ref{#1})} }
\newref{thm}{refcmd = {Theorem \ref{#1}}}
\newref{lem}{refcmd = {Lemma \ref{#1}}}
\newref{prop}{refcmd = {Proposition \ref{#1}}}
\newref{cor}{refcmd = {Corollary \ref{#1}}}

\global\long\def\mb#1{\mathbb{#1}}

\global\long\def\t#1{\tilde{#1}}

\title{An Estimate on the Number of Eigenvalues of a Quasiperiodic Jacobi Matrix of Size $n$ Contained in an Interval of Size $n^{-C}$}
\author{Ilia Binder and Mircea Voda}
\date{}
\maketitle

\begin{abstract}
We consider infinite quasi-periodic Jacobi self-adjoint matrices for
which the three main diagonals are given via values of real analytic
functions on the trajectory of the shift $x\rightarrow x+\omega$.
We assume that the Lyapunov exponent $L(E_{0})$ of the corresponding
Jacobi cocycle satisfies $L(E_{0})\ge\gamma>0$. In this setting we
prove that the number of eigenvalues $E_{j}^{(n)}(x)$ of a submatrix
of size $n$ contained in an interval $I$ centered at $E_{0}$ with
$|I|=n^{-C_{1}}$ does not exceed $\left(\log n\right)^{C_{0}}$ for
any $x$. Here $n\ge n_{0}$, and $n_{0}$, $C_{0}$, $C_{1}$ are
constants depending on $\gamma$ (and the other parameters of the
problem).
\end{abstract}
\let\thefootnote\relax
\let\thefootnote\relax

\section{Introduction}

Denote $\TT:=\mathbb{R}/\mathbb{Z}$ and let $a:\TT\rightarrow\mb R$,
and $b:\TT\rightarrow\mb C$ be real analytic functions, with $b$
not identically zero. Let $\omega\in\left(0,1\right)$ satisfy a (generic)
Diophantine condition of the form
\[
\left\Vert n\omega\right\Vert \ge\frac{C_{\omega}}{n\left(\log n\right)^{\alpha}},
\]
where $\alpha>1$ is fixed. We consider the quasiperiodic Jacobi operator
$H\left(x,\omega\right)$ defined on $l^{2}\left(\mathbb{Z}\right)$
by
\[
\left[H\left(x,\omega\right)\phi\right]\left(k\right)=-b\left(x+\left(k+1\right)\omega\right)\phi\left(k+1\right)-\overline{b\left(x+k\omega\right)}\phi\left(k-1\right)+a\left(x+k\omega\right)\phi\left(k\right).
\]
The important special case given by $b\equiv1$ (Schrödinger operator)
has been studied extensively (see the monograph \cite{MR2100420}).
The study of results that apply to quasiperiodic Jacobi operators
in such a general setting has been launched by the recent work of
Jitomirskaya, Koslover, and Schulteis \cite{MR2563096} and Jitomirskaya
and Marx \cite{MR2825743}. In particular, they studied the extended
Harper's model which corresponds to $a\left(x\right)=2\cos(2\pi x)$,
$b(x)=\lambda_{1}e^{2\pi i(x-\omega/2)}+\lambda_{2}+\lambda_{3}e^{-2\pi i\left(x-\omega/2\right)}$
(see \cite{MR2121278,2010arXiv1010.0751J}). Further motivation for
the study of these operators comes from the general fact that quasiperiodic
Jacobi operators are necessary for the solution of the inverse spectral
problem for discrete quasiperiodic operators of second order, and
for the solution of the Todda Lattice with quasiperiodic initial data.

The main objective of this work is to estimate the number of Dirichlet
eigenvalues of the problem on a finite interval of length $n$ which
fall into a given interval of size $n^{-C}$. This type of estimate
plays a central role in the work of Goldstein and Schlag \cite{MR1847592,MR2438997}.
In our analysis we use many ideas and methods of their work. On the
other hand, as it was noted in \cite{MR2825743}, the singularities
(associated with the zeros of $b$) of the corresponding matrix-functions
introduce considerable technical difficulties. These difficulties
are addressed by using a large deviation theorem for subharmonic functions
(\cite[Theorem 3.8]{MR1847592}) applied to $\log\left|b\right|$,
which will allow us to include the singularities in the exceptional
sets. The derivation of the large deviation estimate for the characteristic
polynomials via the method of \cite{MR2438997} becomes especially
complicated, even if $b$ would have no zeros. We show how to get
around these difficulties by introducing a different derivation which
makes a finer use of the cocyle structure (see the proof of \lemref{considerable-difficulty-lemma}).
Our estimate on the number of eigenvalues also improves on the estimate
in \cite{MR2438997}.

The methods we will employ are complex analytic, so from now we canonically
identify $\TT$ with the unit circle in $\mb C$. It is known that
$a$ and $b$ can be extended to be (complex) analytic on a neighborhood
of $\TT$. Let $\t b\left(z\right):=\overline{b\left(1/\bar{z}\right)}$
denote the analytic extension of $\bar{b}$. We now extend the definition
of $H\left(\cdot,\omega\right)$, to a neighborhood on which both
$a$ and $b$ can be extended, by
\[
\left[H\left(z,\omega\right)\phi\right]\left(k\right)=-b\left(z+\left(k+1\right)\omega\right)\phi\left(k+1\right)-\t b\left(z+k\omega\right)\phi\left(k-1\right)+a\left(z+k\omega\right)\phi\left(k\right).
\]
Note that $H\left(\cdot,\omega\right)$ is not necessarily self-adjoint
off $\TT$. For simplicity we make the notational convention that
$z+k\omega:=z\exp\left(2\pi ik\omega\right)$, for $z\in\mb C$ and
$k\in\mb Z$. 

We consider the finite Jacobi submatrix on $\left[0,n-1\right]$,
denoted by $H^{(n)}\left(z,\omega\right)$, and defined by
\[
\left[\begin{array}{ccccc}
a\left(z\right) & -b\left(z+\omega\right) & 0 & \ldots & 0\\
-\tilde{b}\left(z+\omega\right) & a\left(z+\omega\right) & -b\left(z+2\omega\right) & \ldots & 0\\
\ddots & \ddots & \ddots & \ldots & \vdots\\
0 & \ldots & 0 & -\t b\left(z+\left(n-1\right)\omega\right) & a\left(z+\left(n-1\right)\omega\right)
\end{array}\right].
\]
Let $L\left(E\right)$ be the Lyapunov exponent associated with $H\left(x,\omega\right)$
(see \eqref{Lyapunov}). Our main result is as follows.

\newtheorem*{mainthm}{Main Theorem}

\begin{mainthm}Assume that $E_{0}\in\mb R$ is such that $L\left(E_{0}\right)\ge\gamma>0$.
Then there exist constants $C_{0}=C_{0}\left(\omega\right)$, $C_{1}=C_{1}\left(a,b,E_{0},\omega,\gamma\right)$,
and $n_{0}=n_{0}\left(a,b,E_{0},\omega,\gamma\right)$ such that for
every $x\in\TT$ and $n\ge n_{0}$ the number of eigenvalues for $H^{(n)}\left(x,\omega\right)$
located in $\left\{ E:\,\left|E-E_{0}\right|<n^{-C_{1}}\right\} $
is at most $\left(\log n\right)^{C_{0}}$ and furthermore, for any
$x_{0}\in\TT$ and $n\ge n_{0}$ the number of zeros for $\det\left(H^{(n)}\left(\cdot,\omega\right)-E_{0}\right)$
contained in $\left\{ z:\,\left|z-x_{0}\right|<n^{-1}\right\} $ is
at most $\left(\log n\right)^{C_{0}}$.\end{mainthm}

In the Schrödinger case such estimates and further refinements were
obtained by Goldstein and Schlag (see \cite[Proposition 4.9]{MR2438997}).
In fact we will prove a slightly stronger theorem, \thmref{Number-of-eigenvalues}.

\subsection*{Acknowledgements}

The authors are grateful to Michael Goldstein for suggesting the problem
and for extensive discussions which were instrumental to the completion
of the project. The first author was partially supported by the NSERC
Discovery Grant 5810-2009-298433.

\section{Preliminaries}

We proceed by introducing some notation and giving an overview of
the methods. For $\phi$ satisfying the difference equation $H\left(z,\omega\right)\phi=E\phi$
let $M_{n}$ be the matrix such that
\[
\left[\begin{array}{c}
\phi\left(n\right)\\
\phi\left(n-1\right)
\end{array}\right]=M_{n}\left[\begin{array}{c}
\phi\left(0\right)\\
\phi\left(-1\right)
\end{array}\right],\, n\ge1.
\]
We call $M_{n}$ the fundamental matrix. We clearly have 
\[
M_{n}\left(z\right)=\prod_{j=n-1}^{0}\left(\frac{1}{b\left(z+\left(j+1\right)\omega\right)}\left[\begin{array}{cc}
a\left(z+j\omega\right)-E & -\tilde{b}\left(z+j\omega\right)\\
b\left(z+\left(j+1\right)\omega\right) & 0
\end{array}\right]\right),
\]
for $z$ such that $\prod_{j=1}^{n}b\left(z+j\omega\right)\neq0$.
Note that in order to simplify the notation we suppressed the dependence
on $\omega$ and $E$. We will be doing this throughout the paper
whenever possible. From now on, if needed, we will include the set
on which the matrices $M_{n}$ are not defined in the exceptional
sets. 

It is straightforward to see that
\begin{equation}
M_{n}\left(z\right)=\left[\begin{array}{cc}
f_{n}\left(z\right) & -\frac{\t b\left(z\right)}{b\left(z+\omega\right)}f_{n-1}\left(z+\omega\right)\\
f_{n-1}\left(z\right) & -\frac{\t b\left(z\right)}{b\left(z+\omega\right)}f_{n-2}\left(z+\omega\right)
\end{array}\right],\label{eq:M_n-entries}
\end{equation}
with
\begin{align}
f_{n}\left(z\right) & =\frac{1}{\prod_{j=1}^{n}b\left(z+j\omega\right)}f_{n}^{a}\left(z\right),\label{eq:f...f^a}
\end{align}
where 
\[
f_{n}^{a}\left(z\right)=\det\left[H^{n}\left(z,\omega\right)-E\right].
\]
Since $f_{n}^{a}\left(x,E\right)$ is the characteristic polynomial
of $H^{\left(n\right)}\left(x,\omega\right)$ so it is natural to
estimate the number of eigenvalues by applying Jensen's formula to
$f_{n}^{a}$. For this to work we need upper and lower estimates on
$\log\left|f_{n}^{a}\right|$. These estimates will follow from the
deviations estimates for the fundamental matrix and its entries (see
\thmref{ldt} and \propref{ldt_entries}). 

The main tools for obtaining the deviations estimates for the fundamental
matrix are a deviations estimate for subharmonic functions and the
Avalanche Principle, both of which we recall next. In what follows
$\cA_{\rho}$ will denote the annulus $\left\{ z\in\mathbb{C}:\,\left|z\right|\in\left(1-\rho,1+\rho\right)\right\} $
and $ $we fix $p>\alpha+2$.
\begin{thm}
\label{thm:sh_ldt}(\cite[Theorem 3.8]{MR1847592}) Let $u$ be a
subharmonic function and let 
\[
u\left(z\right)=\int_{\mathbb{C}}\log\left|z-\zeta\right|d\mu\left(\zeta\right)+h\left(z\right)
\]
be its Riesz representation on a neighborhood of $\cA_{\rho}$. If
$\mu\left(\cA_{\rho}\right)+\left\Vert h\right\Vert _{L^{\infty}\left(\cA_{\rho}\right)}\le M$
then for any $\delta>0$ and any positive integer $n$ we have
\[
\mes\left(\left\{ x\in\TT:\left|\sum_{k=1}^{n}u\left(x+k\omega\right)-n\left\langle u\right\rangle \right|>\delta n\right\} \right)<\exp\left(-c_{0}\delta n+r_{n}\right)
\]
where $c_{0}=c_{0}\left(\omega,M,\rho\right)$ and
\[
r_{n}=\begin{cases}
C_{0}\left(\log n\right)^{p} & ,\, n>1\\
C_{0} & ,\, n=1,
\end{cases}
\]
with $C_{0}=C_{0}\left(\omega,p\right)$.
\end{thm}
If $p_{s}/q_{s}$ is a convergent of $\omega$ and $n=q_{s}>1$ then
one can choose $r_{n}=C_{0}\log n$. One can keep this in mind, but
we will make no use of it.
\begin{prop}
\label{prop:avalanche}(Avalanche Principle; \cite[Proposition 3.3]{MR2438997})
Let $A_{1},\ldots,A_{n}$, $n\ge2$, be a sequence of $2\times2$
matrices. If
\[
\max_{1\le j\le n}\left|\det A_{j}\right|\le1,
\]
\[
\min_{1\le j\le n}\left\Vert A_{j}\right\Vert \ge\mu>n,
\]
and
\[
\max_{1\le j<n}\left(\log\left\Vert A_{j+1}\right\Vert +\log\left\Vert A_{j}\right\Vert -\log\left\Vert A_{j+1}A_{j}\right\Vert \right)<\frac{1}{2}\log\mu
\]
then 
\[
\left|\log\left\Vert A_{n}\ldots A_{1}\right\Vert +\sum_{j=2}^{n-1}\log\left\Vert A_{j}\right\Vert -\sum_{j=1}^{n-1}\log\left\Vert A_{j+1}A_{j}\right\Vert \right|<C_{0}\frac{n}{\mu}
\]
with some absolute constant $C_{0}$.
\end{prop}
In \cite{MR1847592} (where $b\equiv1$) one takes advantage of the
fact that $\log\left\Vert M_{n}\left(\cdot\right)\right\Vert $ is
subharmonic (on a neighborhood of $\TT$) and that it is almost invariant
to get a first deviations estimate by using \thmref{sh_ldt}. Next,
this estimate is used to apply the Avalanche Principle, which together
with the almost invariance yields a sharper deviations estimate. Almost
invariance refers to the fact that
\[
\left|\log\left\Vert M_{n}\left(x\right)\right\Vert -\frac{1}{l}\sum_{k=0}^{l-1}\log\left\Vert M_{n}\left(x+k\omega\right)\right\Vert \right|\le Cl,\, x\in\TT.
\]
In our case $\log\left\Vert M_{n}\left(\cdot\right)\right\Vert $
is not necessarily subharmonic, the Avalanche Principle (as stated)
cannot be applied to $M_{n}$, because it possible that $\left|\det M_{n}\right|\nleqslant1$,
and the almost invariance may fail to hold on $\TT$. To work around
these issues it is natural to use the following two matrices associated
with $M_{n}$:
\begin{equation}
\gh_{n}\left(z\right)=\left(\prod_{j=1}^{n}b\left(z+j\omega\right)\right)M_{n}\left(z\right)\label{eq:M^a}
\end{equation}
and
\begin{align}
\fg_{n}\left(z\right) & =\frac{1}{\sqrt{\left|\det M_{n}\left(z\right)\right|}}M_{n}\left(z\right)=\left(\prod_{j=0}^{n-1}\left|\frac{b\left(z+\left(j+1\right)\omega\right)}{\t b\left(z+j\omega\right)}\right|^{1/2}\right)M_{n}\left(z\right)\label{eq:M^u}
\end{align}
$\gh_{n}\left(\cdot\right)$ is analytic and hence $\log\left\Vert \gh_{n}\left(\cdot\right)\right\Vert $
is subharmonic, and $M_{n}^{u}\left(\cdot\right)$ is unimodular (i.e.
$\left|\det M_{n}^{u}\right|=1$). Clearly, we will apply \thmref{sh_ldt}
to $\log\left\Vert M_{n}^{a}\right\Vert $ and the Avalanche Principle
to $M_{n}^{u}$. Note that $\log\left\Vert M_{n}^{a}\left(\cdot\right)\right\Vert $
would be subharmonic even if we had $\bar{b}$ instead of $\tilde{b}$,
however $\t b$ is needed to ensure that $f_{n}^{a}$ is analytic.
Furthermore, if we have $\bar{b}$ instead of $\tilde{b}$ the function
$\log\left|f_{n}^{a}\left(\cdot\right)\right|$ is not necessarily
subharmonic. 

Using \eqref{M_n-entries}, \eqref{M^a}, and \eqref{M^u} it is straightforward
to check that
\begin{equation}
\gh_{n}\left(z\right)=\left[\begin{array}{cc}
\g_{n}\left(z\right) & -\t b\left(z\right)\g_{n-1}\left(z+\omega\right)\\
b\left(z+n\omega\right)\g_{n-1}\left(z\right) & -\t b\left(z\right)b\left(z+n\omega\right)\g_{n-2}\left(z+\omega\right)
\end{array}\right],\label{eq:M^a-entries}
\end{equation}
\begin{equation}
\fg_{n}\left(z\right)=\left[\begin{array}{cc}
\f_{n}\left(z\right) & -\frac{\t b\left(z\right)}{b\left(z+\omega\right)}\left|\frac{b\left(z+\omega\right)}{\t b\left(z\right)}\right|^{1/2}\f_{n-1}\left(z+\omega\right)\\
\left|\frac{b\left(z+n\omega\right)}{\t b\left(z+\left(n-1\right)\omega\right)}\right|^{1/2}\f_{n-1}\left(x\right) & -\frac{\t b\left(z\right)}{b\left(z+\omega\right)}\left|\frac{b\left(z+n\omega\right)b\left(z+\omega\right)}{\t b\left(z+\left(n-1\right)\omega\right)\t b\left(z\right)}\right|^{1/2}\f_{n-2}\left(z+\omega\right)
\end{array}\right],\label{eq:M^u-entries}
\end{equation}
where
\begin{equation}
f_{n}^{u}\left(z\right)=\left(\prod_{j=0}^{n-1}\left|\frac{b\left(z+\left(j+1\right)\omega\right)}{\t b\left(z+j\omega\right)}\right|^{1/2}\right)f_{n}\left(z\right)\label{eq:f...f^u}
\end{equation}
($f_{n}$ and $f_{n}^{a}$ have already been defined).

Let $S_{n}\left(z\right)=\sum_{k=0}^{n-1}\log\left|b\left(z+k\omega\right)\right|$
and $\t S_{n}\left(z\right)=\sum_{k=0}^{n-1}\log\left|\t b\left(z+k\omega\right)\right|$.
From \eqref{M^a} and \eqref{M^u} we get
\begin{equation}
\log\left\Vert M_{n}\left(z\right)\right\Vert =-S_{n}\left(z+\omega\right)+\log\left\Vert \gh_{n}\left(z\right)\right\Vert ,\label{eq:logM...logM^a}
\end{equation}
and
\begin{equation}
\log\left\Vert \fg_{n}\left(z\right)\right\Vert =-\frac{1}{2}\left(\t S_{n}\left(z\right)+S_{n}\left(z+\omega\right)\right)+\log\left\Vert \gh_{n}\left(z\right)\right\Vert .\label{eq:logM^u...logM^a}
\end{equation}
It will be easy to see that these relations together with \thmref{sh_ldt}
applied to $\log\left|b\right|$ and $\log\left|\tilde{b}\right|$
allow us to pass from deviations estimates for $\gh_{n}$ to deviations
estimates for $M_{n}$ and $\fg_{n}$ (see for example \corref{ldt_delta^2...M^u}). 

Even though we will apply the Avalanche Principle to $\fg_{n}$ the
conclusion will also hold for $\gh_{n}$ and $M_{n}$. We will make
this more precise. Let $n=\sum_{j=1}^{m}l_{j}$, $s_{k}=\sum_{j=1}^{k}l_{j}$
where $m,l_{1},\ldots,l_{m}$ are positive integers. We assume that
$s_{0}=0$. By saying that, for example, the conclusion of the Avalanche
Principle applied to $\fg_{n}$ also holds for $\gh_{n}$ we mean
that 
\begin{multline*}
\log\left\Vert \fg_{n}\left(z\right)\right\Vert +\sum_{j=2}^{m-1}\log\left\Vert \fg_{l_{j}}\left(z+s_{j-1}\omega\right)\right\Vert \\
-\sum_{j=1}^{m-1}\log\left\Vert \fg_{l_{j+1}}\left(z+s_{j}\omega\right)\fg_{l_{j}}\left(z+s_{j-1}\omega\right)\right\Vert =\log\left\Vert \gh_{n}\left(z\right)\right\Vert \\
+\sum_{j=2}^{m-1}\log\left\Vert \gh_{l_{j}}\left(z+s_{j-1}\omega\right)\right\Vert -\sum_{j=1}^{m-1}\log\left\Vert \gh_{l_{j+1}}\left(z+s_{j}\omega\right)\gh_{l_{j}}\left(z+s_{j-1}\omega\right)\right\Vert .
\end{multline*}
This follows easily from \eqref{logM^u...logM^a}.

The deviations estimate for $\log\left|f_{n}^{a}\right|$ is just
the John-Nirenberg inequality. The needed $BMO$ norm bound will be
obtained by using the {}``$BMO$ splitting lemma'' \cite[Lemma 2.3]{MR1843776}.
As in the case for the fundamental matrix, we first obtain a rough
estimate (\lemref{ldt_entries_weak}) that allows us to apply the
Avalanche Principle in order to obtain a better estimate. We follow
the approach from \cite{MR2438997} with the notable exception of
the proof of \lemref{considerable-difficulty-lemma} (cf. \cite[Lemma 2.7]{MR2438997}).
This is the only place where the difficulties come not only from the
possible zeroes of $b$ but also from the fact that $b$ is not constant. 

We will obtain a uniform upper bound for $\log\left|f_{n}^{a}\left(\cdot\right)\right|$
on $\TT$ from an uniform upper bound for $\log\left\Vert M_{n}^{a}\left(\cdot\right)\right\Vert $
(\propref{M^a-upper-bound}) and the obvious inequality $\log\left|f_{n}^{a}\left(\cdot\right)\right|$
$\le\log\left\Vert M_{n}^{a}\left(\cdot\right)\right\Vert $. The
proof of \propref{M^a-upper-bound} requires that the deviations estimate
for $\log\left\Vert M_{n}^{a}\right\Vert $ holds on $r\TT$ for $r$
in a neighborhood of $1$. Of course this implies that all the results
leading to the deviations estimate should also hold on $r\TT$. For
simplicity we will prove these estimates on $\TT$, however the proofs
will be such that the generalization from $\TT$ to $r\TT$ is immediate.
To this end the derivations up to \propref{M^a-upper-bound} won't
use the fact that $\t b=\bar{b}$ on $\TT$. However, after that point
we only need the results to hold on $\TT$ and we will make use of
said fact to simplify notation. 

The deviations estimates will rely on the positivity of the Lyapunov
exponent. Let 
\[
L_{n}\left(r\right)=\frac{1}{n}\int_{\mathbb{T}}\log\left\Vert M_{n}\left(rx\right)\right\Vert dx,
\]
\[
L_{n}^{u}\left(r\right)=\frac{1}{n}\int_{\TT}\log\left\Vert M_{n}^{u}\left(rx\right)\right\Vert dx,
\]
\[
L_{n}^{a}\left(r\right)=\frac{1}{n}\int_{\mathbb{T}}\log\left\Vert \gh_{n}\left(rx\right)\right\Vert dx,
\]
\[
D\left(r\right)=\int_{\mathbb{T}}\log\left|b\left(rx\right)\right|dx,
\]
and
\[
\t D\left(r\right)=\int_{\TT}\log\left|\t b\left(rx\right)\right|dx.
\]
When $r=1$ we will omit the $r$ argument, so for example we will
write $L_{n}$ instead of $L_{n}\left(1\right)$. The quantities $L_{n}^{a}\left(r\right)$,
$D\left(r\right)$, and $\t D\left(r\right)$ are finite because the
integrands are subharmonic (and not identically $-\infty$), and $L_{n}\left(r\right)$
is finite because from \eqref{logM...logM^a} we have 
\begin{equation}
L_{n}\left(r\right)=-D\left(r\right)+L_{n}^{a}\left(r\right).\label{eq:L_n...L^a_n}
\end{equation}
By Kingman's subadditive ergodic theorem the following limits exist:
\begin{eqnarray}
L\left(r\right) & = & \lim_{n\rightarrow\infty}L_{n}\left(r\right)=\inf_{n\ge1}L_{n}\left(r\right),\label{eq:Lyapunov}\\
L^{u}\left(r\right) & = & \lim_{n\rightarrow\infty}L_{n}^{u}\left(r\right)=\inf_{n\ge1}L_{n}^{u}\left(r\right),\\
L^{a}\left(r\right) & = & \lim_{n\rightarrow\infty}L_{n}^{a}\left(r\right)=\inf_{n\ge1}L_{n}^{a}\left(r\right).
\end{eqnarray}
$L=L\left(E,\omega\right)$ is called the Lyapunov exponent. From
\eqref{M^u} it can be seen that
\[
L^{u}\left(r\right)=\frac{1}{2}(\t D\left(r\right)-D\left(r\right))+L\left(r\right)
\]
and in particular, since $D=\tilde{D}$, we have $L=L^{u}$. Since
$M_{n}^{u}$ is unimodular it follows that $L_{n}^{u}\left(r\right)\ge0$,
and hence $L^{u}\left(r\right)\ge0$. In particular we have that $L=L^{u}\ge0$. 

Fix $\gamma>0$. From now on we assume that $L\ge\gamma>0$. This
assumption is needed to apply the Avalanche Principle, so in fact
we will use $L^{u}=L\ge\gamma>0$. For the results to hold on $r\TT$,
$r\neq1$, we will need that $r$ is close enough to $1$ so that
$L^{u}\left(r\right)\ge\gamma/2>0$. Note that the results up to \lemref{Ln_L}
don't use the Avalanche Principle and so they hold without the assumption
that $L\ge\gamma>0$.

Henceforth we will assume that $a$ and $b$ are analytic on the closure
of $\cA_{\rho_{0}''}$ with $\rho_{0}''>0$ fixed. We also fix $\rho_{0}$
and $\rho_{0}'$ such that $0<\rho_{0}<\rho_{0}'<\rho_{0}''$. $ $The
reason for this setup is that $\log\left\Vert M_{n}^{a}\left(\cdot\right)\right\Vert $
will have a Riesz representation on $\cA_{\rho_{0}'}$ but we will
be able to get the estimates on the Riesz representation (needed for
\thmref{sh_ldt}) only on $\cA_{\rho_{0}}$. The estimates before
\propref{M^a-upper-bound} will hold on $r\TT$ for every $r\in\left(1-\rho_{0}/2,1+\rho_{0}/2\right)$
(provided $L^{u}\left(r\right)>0$) and the constants can be chosen
uniformly for all such $r$. \propref{M^a-upper-bound} will hold
on $r\TT$ for every $r\in\left(1-\rho_{0}/4,1+\rho_{0}/4\right)$
(provided $L^{u}\left(r\right)>0$).

\section{Estimates for the Fundamental Matrix\label{sec:Estimates-for-Matrix}}

First we prove the almost invariance of $\gh_{n}$ (see \eqref{M^a-almost-invariance}).
The following lemma and its corollaries contain the main estimates
that are needed to deal with the fact that $b$ could have zeros.
If $b$ doesn't have any zeros then all the estimates hold trivially
without exceptional sets and everything goes as in \cite{MR1847592}.

In what follows we will keep track of the dependence of the various
constants on the parameters of our problem. The dependence on $\omega$
will only come up through \thmref{sh_ldt}. In order to simplify the
notation we won't record the dependence on $\rho_{0}$, $\rho_{0}'$,
and $\rho_{0}''$ (except in the lemmas where $\rho_{0}$ appears
in the statement). Dependence on any other quantities is such that
if the quantity takes values in a compact set, then the constant can
be chosen uniformly with respect to that quantity. The main dependence
we are interested in, is that on $\left|E\right|$. We denote by $\left\Vert \cdot\right\Vert _{\infty}$
the $L^{\infty}$ norm on $\cA_{\rho_{0}''}$ and we let $\left\Vert b\right\Vert _{*}=\left\Vert b\right\Vert _{\infty}+\sup_{r\in\left(1-\rho_{0},1+\rho_{0}\right)}$$\left|D\left(r\right)\right|$.
Note that $\left\Vert b\right\Vert _{*}=\left\Vert \tilde{b}\right\Vert _{*}$.
\begin{lem}
There exist constants $\lambda_{0}=\lambda_{0}\left(\left\Vert a\right\Vert _{\infty},\left\Vert b\right\Vert _{*},\left|E\right|,\omega\right)$
and $c_{0}=c_{0}(\left\Vert b\right\Vert _{*}$$,$ $\omega)$ such
that the following inequalities hold for any positive integer $l$
and any $\lambda\ge\lambda_{0}$ up to a set (independent of $E$)
of measure less than $\exp\left(-c_{0}\lambda l\right)$:
\begin{gather}
\left|\log\left\Vert \gh_{l}\left(x\right)\right\Vert \right|\le\lambda l\label{eq:M^a-bound}\\
\left|\log\left\Vert \gh_{l}\left(x\right)^{-1}\right\Vert \right|\le\lambda l.\label{eq:M^a^(-1)-bound}
\end{gather}
\end{lem}
\begin{proof}
There exists a constant $C=C\left(\left\Vert a\right\Vert _{\infty},\left\Vert b\right\Vert _{\infty},\left|E\right|\right)>0$
such that 
\[
\log\left\Vert \gh_{l}\left(x\right)\right\Vert \le Cl
\]
for all $x$. On the other hand
\[
\left\Vert \gh_{l}\left(x\right)\right\Vert \ge\left|\det\gh_{l}\left(x\right)\right|^{1/2}=\prod_{j=0}^{l-1}\left|\t b\left(x+j\omega\right)b\left(x+\left(j+1\right)\omega\right)\right|^{1/2}
\]
for all $x$. Hence
\begin{equation}
\frac{\t S_{l}\left(x\right)+S_{l}\left(x+\omega\right)}{2}\le\log\left\Vert \gh_{l}\left(x\right)\right\Vert \le Cl\label{eq:T-bounds-everywhere}
\end{equation}
for all $x$. From \thmref{sh_ldt} we can conclude that for any $\lambda'>0$
we have 
\[
-2\lambda'l\le\left(\frac{\t D+D}{2}-\lambda'\right)l\le\log\left\Vert \gh_{l}\left(x\right)\right\Vert \le Cl\le2\lambda'l
\]
up to a set not exceeding $2\exp\left(-c\lambda'l+r_{l}\right)$ in
measure, provided 
\[
\lambda'\ge\max\left\{ -\left(\t D+D\right),C\right\} /2.
\]
By setting $\lambda=2\lambda'$ and choosing $\lambda_{0}\ge\max\left\{ -\left(\t D+D\right),C\right\} $
we have that \eqref{M^a-bound} holds up to a set of measure not exceeding
$2\exp\left(-c\lambda l+r_{l}\right)$. Finally, it is easy to see
that by choosing $\lambda_{0}$ such that
\[
\lambda_{0}\ge\frac{2}{c}\sup_{l\ge1}\frac{\log2+r_{l}}{l}
\]
we have
\[
2\exp\left(-c\lambda l+r_{l}\right)\le\exp\left(-\frac{c}{2}\lambda l\right),\,\lambda\ge\lambda_{0},\, l\ge1.
\]
This concludes the proof of \eqref{M^a-bound}.

Since for almost every $x$ we have
\begin{multline*}
\left[\begin{array}{cc}
a\left(x+j\omega\right)-E & -\tilde{b}\left(x+j\omega\right)\\
b\left(x+\left(j+1\right)\omega\right) & 0
\end{array}\right]^{-1}=\\
\frac{1}{\tilde{b}\left(x+j\omega\right)b\left(x+\left(j+1\right)\omega\right)}\left[\begin{array}{cc}
0 & \tilde{b}\left(x+j\omega\right)\\
-b\left(x+\left(j+1\right)\omega\right) & a\left(x+j\omega\right)-E
\end{array}\right]
\end{multline*}
it is straightforward to see that there exists a constant $C=C\left(\left\Vert a\right\Vert _{\infty},\left\Vert b\right\Vert _{\infty},\left|E\right|\right)$
such that
\[
-\frac{\t S_{l}\left(x\right)+S_{l}\left(x+\omega\right)}{2}\le\log\left\Vert \gh_{l}\left(x\right)^{-1}\right\Vert \le Cl-\t S_{l}\left(x\right)-S_{l}\left(x+\omega\right)
\]
for almost every $x$. Now \eqref{M^a^(-1)-bound} follows in the
same way as \eqref{M^a-bound}. Note that the exceptional set comes
from $\t S_{l}\left(x\right)+S_{l}\left(x+\omega\right)$ and is thus
independent of $E$.
\end{proof}
The same type of estimates can be obtained now for $M_{n}$ and $\fg_{n}$.
We just record one of the estimates that will be needed later.
\begin{cor}
\label{cor:M^u^-1bound} There exist constants $\lambda_{0}=\lambda_{0}\left(\left\Vert a\right\Vert _{\infty},\left\Vert b\right\Vert _{*},\left|E\right|,\omega\right)$
and $c_{0}=c_{0}\left(\left\Vert b\right\Vert _{*},\omega\right)$
such that 
\[
\left|\log\left\Vert \fg_{l}\left(x\right)^{-1}\right\Vert \right|\le\lambda l
\]
holds for any positive integer $l$ and any $\lambda\ge\lambda_{0}$
up to a set of measure less than $\exp\left(-c_{0}\lambda l\right)$.\end{cor}
\begin{proof}
From \eqref{M^u} we have 
\[
\log\left\Vert \fg_{l}\left(x\right)^{-1}\right\Vert =\frac{1}{2}\left(\t S_{l}\left(x\right)+S_{l}\left(x+\omega\right)\right)+\log\left\Vert \gh_{l}\left(x\right)^{-1}\right\Vert .
\]
Using \thmref{sh_ldt} and \eqref{M^a^(-1)-bound} we get
\[
-3\lambda'l\le\left(\frac{\t D+D}{2}-2\lambda'\right)l\le\log\left\Vert \fg_{l}\left(x\right)^{-1}\right\Vert \le\left(\frac{\t D+D}{2}+2\lambda'\right)l\le3\lambda'l
\]
up to a set of measure less than $2\exp\left(-c_{1}\lambda'l+r_{l}\right)+\exp\left(-c_{2}\lambda'l\right)\le\exp\left(-c3\lambda'l\right)$
provided $\lambda'$ is large enough. Now we can take $\lambda=3\lambda'$. \end{proof}
\begin{cor}
There exist constants $\lambda_{0}=\lambda_{0}\left(\left\Vert a\right\Vert _{\infty},\left\Vert b\right\Vert _{*},\left|E\right|,\omega\right)$
and $c_{0}=c_{0}\left(\left\Vert b\right\Vert _{*},\omega\right)$
such that the following inequalities hold for any positive integers
$l$ and $n$, and any $\lambda\ge\lambda_{0}$ up to a set (depending
on $n$) of measure less than $\exp\left(-c_{0}\lambda l\right)$:
\begin{gather}
\left|\log\left\Vert \gh_{l}\left(x\right)\right\Vert -lL_{l}^{a}\right|\le\lambda l\label{eq:M^a-L^a-bound}\\
\left|\log\left\Vert \gh_{n+l}\left(x\right)\right\Vert -\log\left\Vert \gh_{n}\left(x\right)\right\Vert \right|\le\lambda l\label{eq:M^a_{n+l}-M^a_{n}-bound}\\
\left|\log\left\Vert \gh_{n}\left(x+l\omega\right)\right\Vert -\log\left\Vert \gh_{n}\left(x\right)\right\Vert \right|\le\lambda l\label{eq:M^a(x+l)-M^a(x)-bound}\\
\left|\log\left\Vert \gh_{n}\left(x\right)\right\Vert -\frac{1}{l}\sum_{k=0}^{l-1}\log\left\Vert \gh_{n}\left(x+k\omega\right)\right\Vert \right|\le\lambda l.\label{eq:M^a-almost-invariance}
\end{gather}
\end{cor}
\begin{proof}
By integrating \eqref{T-bounds-everywhere} we get
\begin{equation}
\frac{\t D+D}{2}\le L_{l}^{a}\le C.\label{eq:L^a_l-bounds}
\end{equation}
This and \eqref{M^a-bound} imply \eqref{M^a-L^a-bound}.

We have
\[
\gh_{n+l}\left(x\right)=\gh_{l}\left(x+n\omega\right)\gh_{n}\left(x\right),
\]
hence 
\[
-\log\left\Vert \gh_{l}\left(x+n\omega\right)^{-1}\right\Vert \le\log\left\Vert \gh_{n+l}\left(x\right)\right\Vert -\log\left\Vert \gh_{n}\left(x\right)\right\Vert \le\log\left\Vert \gh_{l}\left(x+n\omega\right)\right\Vert 
\]
for almost every $x$. Now \eqref{M^a_{n+l}-M^a_{n}-bound} follows
by \eqref{M^a-bound} and \eqref{M^a^(-1)-bound}.

From the fact that
\[
\gh_{n}\left(x+l\omega\right)\gh_{l}\left(x\right)=\gh_{l}\left(x+n\omega\right)\gh_{n}\left(x\right)
\]
we conclude that
\begin{multline*}
-\log\left\Vert \gh_{l}\left(x+n\omega\right)^{-1}\right\Vert -\log\left\Vert \gh_{l}\left(x\right)\right\Vert \le\log\left\Vert \gh_{n}\left(x+l\omega\right)\right\Vert -\log\left\Vert \gh_{n}\left(x\right)\right\Vert \\
\le\log\left\Vert \gh_{l}\left(x+n\omega\right)\right\Vert +\log\left\Vert \gh_{l}\left(x\right)^{-1}\right\Vert 
\end{multline*}
for almost every $x$. Now \eqref{M^a(x+l)-M^a(x)-bound} also follows
by \eqref{M^a-bound} and \eqref{M^a^(-1)-bound}.

Let $\lambda\ge\lambda_{0}$. Then for $k=1,\ldots,l-1$ we have $\lambda l/k>\lambda_{0}$,
so by \eqref{M^a(x+l)-M^a(x)-bound} we get
\[
\left|\log\left\Vert \gh_{n}\left(x+k\omega\right)\right\Vert -\log\left\Vert \gh_{n}\left(x\right)\right\Vert \right|\le\left(\frac{\lambda l}{k}\right)k=\lambda l
\]
up to a set of measure less than $\exp\left(-c\lambda l\right)$.
Summing over $k=0,\ldots,l-1$ and dividing by $l$ we get that \eqref{M^a-almost-invariance}
holds up to a set of measure less than $l\exp\left(-c\lambda l\right)$.
Finally, note that $l\exp\left(-c\lambda l\right)<\exp\left(-c\lambda l/2\right),\, l\ge1$
if $\lambda$ is large enough. This concludes the proof.
\end{proof}
Next we provide bounds on the Riesz representation of $\log\left\Vert \gh_{n}\left(\cdot\right)\right\Vert $,
which are needed to ensure that the constants we will get from \thmref{sh_ldt}
don't depend on $n$.
\begin{lem}
\label{lem:M^a-Riesz_bounds} Let 
\[
\frac{1}{n}\log\left\Vert \gh_{n}\left(z\right)\right\Vert =\int_{\cA_{\rho_{0}'}}\log\left|z-\zeta\right|d\mu_{n}\left(\zeta\right)+h_{n}\left(z\right)
\]
be the Riesz representation on $\cA_{\rho_{0}'}$. There exists a
constant $C_{0}=C_{0}(\left\Vert a\right\Vert _{\infty},\left\Vert b\right\Vert _{*},$
$\left|E\right|,\rho_{0},\rho_{0}',\rho_{0}'')$ such that
\[
\mu_{n}\left(\cA_{\rho_{0}}\right)+\left\Vert h_{n}\right\Vert _{L^{\infty}\left(\cA_{\rho_{0}}\right)}\le C_{0}.
\]
\end{lem}
\begin{proof}
Let $u_{n}\left(z\right)=\log\left\Vert M_{n}^{a}\left(z\right)\right\Vert /n$
and $T_{n}=\sup_{\cA_{\rho_{0}''}}u_{n}$. From \cite[Lemma 2.2]{MR2438997}
we have that 
\begin{multline*}
\mu_{n}\left(\cA_{\rho_{0}}\right)\le\mu_{n}\left(\cA_{\rho_{0}'}\right)\le C\left(\rho_{0}',\rho_{0}''\right)\left(T_{n}-\sup_{\cA_{\rho_{0}'}}u_{n}\right)\\
\le C\left(T_{n}-\sup_{\TT}u_{n}\right)\le C\left(T_{n}-L_{n}^{a}\right)
\end{multline*}
and
\begin{multline*}
\left\Vert h_{n}\right\Vert _{L^{\infty}\left(\cA_{\rho_{0}}\right)}\le\left\Vert h_{n}-\sup_{\cA_{\rho_{0}'}}u_{n}\right\Vert _{L^{\infty}\left(\cA_{\rho_{0}}\right)}+\sup_{\cA_{\rho_{0}'}}u_{n}\\
\le C\left(\rho_{0},\rho_{0}',\rho_{0}''\right)\left(T_{n}-\sup_{\cA_{\rho_{0}'}}u_{n}\right)+T_{n}\le C\left(T_{n}-L_{n}^{a}\right)+T_{n}.
\end{multline*}
The conclusion now follows from the fact that there exists a constant
$C=C(\left\Vert a\right\Vert _{\infty},$ $\left\Vert b\right\Vert _{\infty},\left|E\right|,\rho_{0}'')$
such that $T_{n}\le C$, and from \eqref{L^a_l-bounds}.
\end{proof}
Now we can prove the first deviations estimate. 
\begin{lem}
\label{lem:ldt_delta^2} Let $\delta_{0}>0$. For any $\delta\in\left(0,\delta_{0}\right)$
and any integer $n>1$ we have 
\[
\mes\left\{ x\in\TT:\left|\log\left\Vert \gh_{n}\left(x\right)\right\Vert -nL_{n}^{a}\right|>n\delta\right\} <\exp\left(-c_{0}n\delta^{2}+C_{0}\left(\log n\right)^{p}\right)
\]
where $c_{0}=c_{0}\left(\left\Vert a\right\Vert _{\infty},\left\Vert b\right\Vert _{*},\left|E\right|,\omega,\delta_{0}\right)$
and $C_{0}=C_{0}\left(\left\Vert a\right\Vert _{\infty},\left\Vert b\right\Vert _{*},\left|E\right|,\omega,p,\delta_{0}\right)$. \end{lem}
\begin{proof}
We have
\begin{multline}
\mes\left\{ x\in\TT:\left|\log\left\Vert \gh_{n}\left(x\right)\right\Vert -nL_{n}^{a}\right|>n\delta\right\} \le\\
\mes\left\{ x\in\TT:\left|\frac{1}{n}\log\left\Vert \gh_{n}\left(x\right)\right\Vert -\frac{1}{l}\sum_{k=0}^{l-1}\frac{1}{n}\log\left\Vert \gh_{n}\left(x+k\omega\right)\right\Vert \right|>\frac{\delta}{2}\right\} \\
+\mes\left\{ x\in\TT:\left|\frac{1}{l}\sum_{k=0}^{l-1}\frac{1}{n}\log\left\Vert \gh_{n}\left(x+k\omega\right)\right\Vert -L_{n}^{a}\right|>\frac{\delta}{2}\right\} .\label{eq:M^a-L^a-decomp}
\end{multline}
The conclusion will follow by estimating the two quantities on the
right-hand side of the above inequality.

From \eqref{M^a-almost-invariance} we get
\[
\left|\log\left\Vert \gh_{n}\left(x\right)\right\Vert -\frac{1}{l}\sum_{k=0}^{l-1}\log\left\Vert \gh_{n}\left(x+k\omega\right)\right\Vert \right|\le C_{1}l
\]
up to a set not exceeding $\exp\left(-cl\right)$ in measure. Let
$l=\left[\delta n/2C_{1}\right]+1$ . We have
\[
\frac{\delta}{2}<\frac{C_{1}l}{n}
\]
so we get
\[
\left|\frac{1}{n}\log\left\Vert \gh_{n}\left(x\right)\right\Vert -\frac{1}{l}\sum_{k=0}^{l-1}\frac{1}{n}\log\left\Vert \gh_{n}\left(x+k\omega\right)\right\Vert \right|\le\frac{\delta}{2}
\]
for all $x$ except for a set of measure less than $\exp\left(-cl\right)$.
Hence 
\begin{multline*}
\mes\left\{ x\in\TT:\left|\frac{1}{n}\log\left\Vert \gh_{n}\left(x\right)\right\Vert -\frac{1}{l}\sum_{k=0}^{l-1}\frac{1}{n}\log\left\Vert \gh_{n}\left(x+k\omega\right)\right\Vert \right|>\frac{\delta}{2}\right\} \\
<\exp\left(-cl\right)<\exp\left(-c_{1}\delta n\right),
\end{multline*}
where $c_{1}=c/(2C_{1})$. 

From \thmref{sh_ldt} we have
\begin{multline*}
\mes\left\{ x\in\TT:\left|\frac{1}{l}\sum_{k=0}^{l-1}\frac{1}{n}\log\left\Vert \gh_{n}\left(x+k\omega\right)\right\Vert -L_{n}^{a}\right|>\frac{\delta}{2}\right\} \\
<\exp\left(-c\frac{\delta}{2}l+C\left(\log l\right)^{p}\right)<\exp\left(-c_{2}\delta^{2}n+C'\left(\log n\right)^{p}\right).
\end{multline*}
Recall that \lemref{M^a-Riesz_bounds} ensures that $c$ and $C$
don't depend on $n$.

Now \eqref{M^a-L^a-decomp} becomes
\begin{multline*}
\mes\left\{ x\in\TT:\left|\log\left\Vert \gh_{n}\left(x\right)\right\Vert -nL_{n}^{a}\right|>n\delta\right\} \\
<\exp\left(-c_{1}\delta n\right)+\exp\left(-c_{2}\delta^{2}n+C'\left(\log n\right)^{p}\right)\\
<2\exp\left(-c\delta^{2}n+C'\left(\log n\right)^{p}\right)<\exp\left(-c\delta^{2}n+C''\left(\log n\right)^{p}\right),
\end{multline*}
where $c=c(c_{1},c_{2},\delta_{0})$. This concludes the proof.
\end{proof}
The same proof yields that for $\delta\ge\delta_{0}$ we have 
\[
\mes\left\{ x\in\TT:\left|\log\left\Vert \gh_{n}\left(x\right)\right\Vert -nL_{n}^{a}\right|>n\delta\right\} <\exp\left(-c_{0}n\delta+C_{0}\left(\log n\right)^{p}\right).
\]
 For $\delta_{0}$ large enough, this just follows from \eqref{M^a-L^a-bound}.
Also note that to get an estimate when $n=1$ one just needs to apply
\thmref{sh_ldt}.

The same type of estimate holds for $M_{n}^{u}$ and $M_{n}$. We
state it only for $M_{n}^{u}$ since this is all we need.
\begin{cor}
\label{cor:ldt_delta^2...M^u}Let $\delta_{0}>0$. For any $\delta\in\left(0,\delta_{0}\right)$
and any integer $n>1$ we have 
\[
\mes\left\{ x\in\TT:\left|\log\left\Vert \fg_{n}\left(x\right)\right\Vert -nL_{n}^{u}\right|>n\delta\right\} <\exp\left(-c_{0}n\delta^{2}+C_{0}\left(\log n\right)^{p}\right)
\]
where $c_{0}=c_{0}\left(\left\Vert a\right\Vert _{\infty},\left\Vert b\right\Vert _{*},\left|E\right|,\omega,\delta_{0}\right)$
and $C_{0}=C_{0}\left(\left\Vert a\right\Vert _{\infty},\left\Vert b\right\Vert _{*},\left|E\right|,\omega,p,\delta_{0}\right)$. \end{cor}
\begin{proof}
Using \eqref{logM^u...logM^a} we easily get
\begin{multline*}
\mes\left\{ x\in\TT:\left|\log\left\Vert \fg_{n}\left(x\right)\right\Vert -nL_{n}^{u}\right|>n\delta\right\} \\
\le\mes\left\{ x\in\TT:\left|\log\left\Vert \gh_{n}\left(x\right)\right\Vert -nL_{n}^{a}\right|>\frac{n\delta}{2}\right\} \\
+\mes\left\{ x\in\TT:\left|\t S_{n}\left(x\right)-n\t D\right|>\frac{n\delta}{2}\right\} +\mes\left\{ x\in\TT:\left|S_{n}\left(x+\omega\right)-nD\right|>\frac{n\delta}{2}\right\} .
\end{multline*}
The conclusion now follows from \lemref{ldt_delta^2} and \thmref{sh_ldt}.
\end{proof}
The next step is to make use of the Avalanche Principle to improve
the previous estimate. The following lemma is the most general application
of the Avalanche Principle that suits our purposes.
\begin{lem}
\label{lem:AppliedAP}Let $n>1$ be an integer and $n=\sum_{j=1}^{m}l_{j}$
where $l_{j}$ are positive integers such that $l\le l_{j}\le3l$,
with $l=l\left(n\right)$ a real number. Let $A_{j}\left(x\right)=A_{j}\left(x,n\right)$
be $2\times2$ matrices for $x\in\TT$, and let $L_{k}$, $k\ge1$
be a sequence of real numbers. If
\[
l>\frac{2}{\gamma}\log n,
\]
\[
L_{l_{j}}-L_{l_{j}+l_{j+1}}\le\frac{\gamma}{100},\, L_{l_{j+1}}-L_{l_{j}+l_{j+1}}\le\frac{\gamma}{100},\, j=1,\ldots,m-1
\]
\[
\max_{1\le j\le m}\left|\det A_{j}\left(x\right)\right|\le1,\, a.e.\, x\in\TT,
\]
\[
\mes\left\{ x\in\TT:\left|\frac{1}{l_{j}}\log\left\Vert A_{j}\left(x\right)\right\Vert -L_{l_{j}}\right|>\frac{\gamma}{100}\right\} \le\exp\left(-c_{0}l_{j}^{\sigma}\right),\, j=1,\ldots,m,
\]
and
\begin{multline*}
\mes\left\{ x\in\TT:\left|\frac{1}{l_{j}+l_{j+1}}\log\left\Vert A_{j+1}\left(x\right)A_{j}\left(x\right)\right\Vert -L_{l_{j}+l_{j+1}}\right|>\frac{\gamma}{100}\right\} \\
\le\exp\left(-c_{0}\left(l_{j}+l_{j+1}\right)^{\sigma}\right),\, j=1,\ldots,m-1,
\end{multline*}
then there exists an absolute constant $C_{0}$ such that
\begin{multline*}
\left|\log\left\Vert A_{m}\left(x\right)\ldots A_{1}\left(x\right)\right\Vert +\sum_{j=2}^{m-1}\log\left\Vert A_{j}\left(x\right)\right\Vert -\sum_{j=1}^{m-1}\log\left\Vert A_{j+1}\left(x\right)A_{j}\left(x\right)\right\Vert \right|\\
<C_{0}m\exp\left(-\frac{\gamma}{2}l\right)<C_{0}\frac{1}{l}
\end{multline*}
up to a set of measure less than $3n\exp\left(-c_{0}l^{\sigma}\right)$.\end{lem}
\begin{proof}
Let $\mu=\exp\left(l\gamma/2\right)$. We have
\[
\min_{1\le j\le m}\left\Vert A_{j}\left(x\right)\right\Vert \ge\min_{1\le j\le m}\exp\left(l_{j}L_{l_{j}}-\frac{\gamma}{100}\right)>\exp\left(l\gamma/2\right)=\mu>n
\]
and
\begin{multline*}
\max_{1\le j<m-1}\left[\log\left\Vert A_{j+1}\left(x\right)\right\Vert +\log\left\Vert A_{j}\left(x\right)\right\Vert -\log\left\Vert A_{j+1}\left(x\right)A_{j}\left(x\right)\right\Vert \right]\\
\le l_{j+1}\left(L_{l_{j+1}}+\frac{\gamma}{100}\right)+l_{j}\left(L_{l_{j}}+\frac{\gamma}{100}\right)-\left(l_{j}+l_{j+1}\right)\left(L_{l_{j+1}+l_{j}}-\frac{\gamma}{100}\right)\\
=l_{j+1}\left(L_{l_{j+1}}-L_{l_{j+1}+l_{j}}+\frac{2\gamma}{100}\right)+l_{j}\left(L_{l_{j}}-L_{l_{j+1}+l_{j}}+\frac{2\gamma}{100}\right)\\
<6l\frac{3\gamma}{100}<\frac{\gamma l}{4}=\frac{1}{2}\log\mu
\end{multline*}
up to a set of measure $3m\exp\left(-c_{0}l^{\sigma}\right)<3n\exp\left(-c_{0}l^{\sigma}\right)$.
The conclusion follows from the Avalanche Principle and the fact that
$m/\mu<1/l$.
\end{proof}
As mentioned before, it is important for us that the constants in
the deviations estimate can be chosen uniformly for $E$ in a compact
set. For this we need to provide a bound for $L_{n}^{u}-L^{u}$ that
holds for all $E$ in a compact set. First we state a simple estimate
that we will use to deal with the integrals over the exceptional sets
for our functions.
\begin{lem}
\label{lem:Integral-lemma}Let $f$ be a measurable function defined
on $\TT$ such that for any $\delta\ge\delta_{0}$ we have that $\left|f\left(x\right)\right|\le\delta$
up to a set of measure less than $\exp\left(-c_{0}\delta\right)$.
Then $\left\Vert f\right\Vert _{L^{2}\left(\TT\right)}\le C_{0}$,
where $C_{0}=C_{0}\left(c_{0},\delta_{0}\right)$.
\end{lem}
\begin{lem}\label{lem:Ln_L} For any integer $n>1$ we have
\[
0\le L_{n}-L=L_{n}^{u}-L^{u}=L_{n}^{a}-L^{a}<C_{0}\frac{\left(\log n\right)^{2}}{n}
\]
where $C_{0}=C_{0}\left(\left\Vert a\right\Vert _{\infty},\left\Vert b\right\Vert _{*},\left|E\right|,\omega,\gamma\right)$.\end{lem}
\begin{proof}
It is sufficient to get the estimate for large $n$. We will tacitly
assume that $n$ is large enough for our estimates to hold. We should
keep in mind that the choice of large $n$ should be uniform for $E$
in a bounded set. 

It is easy to see that the conclusion follows if we have
\begin{equation}
\left|L_{2n}^{a}-L_{n}^{a}\right|\le C\frac{\left(\log n\right)^{2}}{n}.\label{eq:Ln-L2n}
\end{equation}
Since we have 
\begin{align*}
\left|L_{2n}^{a}-L_{n}^{a}\right| & =\left|\int_{\TT}\frac{\log\left\Vert \gh_{2n}\left(x\right)\right\Vert -\log\left\Vert \gh_{n}\left(x+n\omega\right)\right\Vert -\log\left\Vert \gh_{n}\left(x\right)\right\Vert }{2n}dx\right|,
\end{align*}
it will be sufficient to prove that
\begin{equation}
\left|\log\left\Vert \gh_{2n}\left(x\right)\right\Vert -\log\left\Vert \gh_{n}\left(x+n\omega\right)\right\Vert -\log\left\Vert \gh_{n}\left(x\right)\right\Vert \right|\le C_{1}\left(\log n\right)^{2}\label{eq:lem:Ln_L-desired-estimate}
\end{equation}
up to a set not exceeding $C_{2}n^{-1}$ in measure. Indeed, from
\eqref{M^a-bound} it follows that for $\delta\ge\delta_{0}$ we have
\[
\left|\frac{\log\left\Vert \gh_{2n}\left(x\right)\right\Vert -\log\left\Vert \gh_{n}\left(x+n\omega\right)\right\Vert -\log\left\Vert \gh_{n}\left(x\right)\right\Vert }{2n}\right|\le\delta
\]
up to a set not exceeding $\exp\left(-c_{1}\delta2n\right)+2\exp\left(-c_{1}\delta n\right)<\exp\left(-c\delta n\right)$
in measure, and by using \eqref{lem:Ln_L-desired-estimate} and \lemref{Integral-lemma}
we get 
\begin{multline*}
\left|L_{2n}^{a}-L_{n}^{a}\right|\le\int_{\TT}\left|\frac{\log\left\Vert \gh_{2n}\left(x\right)\right\Vert -\log\left\Vert \gh_{n}\left(x+n\omega\right)\right\Vert -\log\left\Vert \gh_{n}\left(x\right)\right\Vert }{2n}\right|dx\\
\le C_{1}\left(\log n\right)^{2}+C_{3}\sqrt{C_{2}n^{-1}}\le C\left(\log n\right)^{2}.
\end{multline*}

Now we check that the sufficient condition \eqref{lem:Ln_L-desired-estimate}
holds. Let $l=\left[C_{l}\log n\right]$ and $m=\left[n/l\right]$.
If $C_{l}$ is sufficiently large we have that $l>2\log n/\gamma$
and $3n\exp\left(-cl\right)<n^{-1}$. We want to choose $C_{l}$ so
that $L_{l}^{u}-L_{2l}^{u}\le\gamma/100$ and $C_{l}\le C$ (note
that without the bound, such $C_{l}$ obviously exists). Suppose that
$L_{2^{j}l}^{u}-L_{2^{j+1}l}^{u}>\frac{\gamma}{100}$ for $j\ge0$.
Then using \eqref{L^a_l-bounds} we get
\[
C-\frac{\t D+D}{2}\ge L_{l}^{u}-L_{2^{j+1}l}^{u}>\frac{j\gamma}{100}.
\]
This shows that by eventually replacing $l$ with $2^{j}l$ with some
\[
j<100\left(2C-\t D-D\right)/\gamma
\]
 we will have $L_{l}^{u}-L_{2l}^{u}\le\gamma/100$, and the corresponding
$C_{l}$ will be bounded. Using \corref{ldt_delta^2...M^u} and \lemref{AppliedAP}
we get$ $ 
\begin{equation}
\left|\log\left\Vert \gh_{ml}\left(x\right)\right\Vert +\sum_{j=1}^{m-2}\log\left\Vert \gh_{l}\left(x+jl\omega\right)\right\Vert -\sum_{j=0}^{m-2}\log\left\Vert \gh_{2l}\left(x+jl\omega\right)\right\Vert \right|<C\label{eq:Tml-avalanche}
\end{equation}
up to a set not exceeding $n^{-1}$ in measure, and analogous estimates
for $\log\Vert\gh(x+$ $ml\omega)\Vert$ and $\log\left\Vert \gh_{2ml}\left(x\right)\right\Vert $.
Recall that we apply the Avalanche Principle to $M_{n}^{u}$ but the
conclusion also holds for $M_{n}^{a}$. Note that we need to have
$m\ge2$. This clearly holds for large enough $n$ depending on $C_{l}$.
This can be done uniformly for $E$ in a bounded set because of our
bound on $C_{l}$. Putting these estimates together we get
\begin{multline}
\Big|\log\left\Vert \gh_{2ml}\left(x\right)\right\Vert -\log\left\Vert \gh_{ml}\left(x+ml\omega\right)\right\Vert -\log\left\Vert \gh_{ml}\left(x\right)\right\Vert \\
+\log\left\Vert \gh_{l}\left(x+\left(m-1\right)l\omega\right)\right\Vert +\log\left\Vert \gh_{l}\left(x+ml\omega\right)\right\Vert \\
-\log\left\Vert \gh_{2l}\left(x+\left(m-1\right)l\omega\right)\right\Vert \Big|<C\label{eq:T2ml-Tml-Tml-stuff}
\end{multline}
up to a set not exceeding $Cn^{-1}$ in measure. By \eqref{M^a-bound}
we have that $\left|\log\left\Vert \gh_{l}\left(x\right)\right\Vert \right|\le C\log n$
up to a set not exceeding $n^{-1}$ in measure. From this, similar
estimates, and \eqref{T2ml-Tml-Tml-stuff} we get 
\begin{equation}
\left|\log\left\Vert \gh_{2ml}\left(x\right)\right\Vert -\log\left\Vert \gh_{ml}\left(x+ml\omega\right)\right\Vert -\log\left\Vert \gh_{ml}\left(x\right)\right\Vert \right|<C\log n\label{eq:T2ml-Tml-Tml}
\end{equation}
up to a set not exceeding $Cn^{-1}$ in measure.

From \eqref{M^a_{n+l}-M^a_{n}-bound} we get that for sufficiently
large $\delta$ we have
\[
\left|\log\left\Vert \gh_{n}\left(x\right)\right\Vert -\log\left\Vert \gh_{ml}\left(x\right)\right\Vert \right|\le\delta\left(n-ml\right)
\]
up to a set not exceeding $\exp\left(-c\delta\left(n-ml\right)\right)$
in measure. We can choose $\delta>(\log n)/c$ to conclude that
\[
\left|\log\left\Vert \gh_{n}\left(x\right)\right\Vert -\log\left\Vert \gh_{ml}\left(x\right)\right\Vert \right|\le C\left(\log n\right)^{2}
\]
up to a set not exceeding $n^{-1}$ in measure. From this, similar
estimates (using \eqref{M^a_{n+l}-M^a_{n}-bound} and \eqref{M^a(x+l)-M^a(x)-bound}),
and \eqref{T2ml-Tml-Tml} we can conclude that
\[
\left|\log\left\Vert \gh_{2n}\left(x\right)\right\Vert -\log\left\Vert \gh_{n}\left(x+n\omega\right)\right\Vert -\log\left\Vert \gh_{n}\left(x\right)\right\Vert \right|<C\left(\log n\right)^{2}
\]
up to a set not exceeding $Cn^{-1}$ in measure. Thus we proved \eqref{lem:Ln_L-desired-estimate}
and this concludes the proof.
\end{proof}
The bound from the previous lemma can be improved, as in \cite[Theorem 5.1]{MR1847592},
to $L_{n}-L\le C_{0}/n$. However, we won't need this better bound
in this paper.

Now we are able to prove the improved version of the deviations estimate
(cf. \cite[Theorem 7.1]{MR1847592}).
\begin{thm}
\label{thm:ldt}For any $\delta>0$ and any integer $n>1$ we have
\[
\mes\left\{ x\in\TT:\left|\log\left\Vert \gh_{n}\left(x\right)\right\Vert -nL_{n}^{a}\right|>\delta n\right\} <\exp\left(-c_{0}\delta n+C_{0}\left(\log n\right)^{p}\right)
\]
where $c_{0}=c_{0}\left(\left\Vert a\right\Vert _{\infty},\left\Vert b\right\Vert _{*},\left|E\right|,\omega,\gamma\right)$
and $C_{0}=C_{0}\left(\left\Vert a\right\Vert _{\infty},\left\Vert b\right\Vert _{*},\left|E\right|,\omega,\gamma,p\right)$.
The same estimate, with possibly different constants, holds with $L^{a}$
instead of $L_{n}^{a}$. \end{thm}
\begin{proof}
First note that due to \eqref{M^a-L^a-bound} we just need to check
the estimate for $\delta<\delta_{0}$. Furthermore, note that the
estimate is trivial if $-c_{0}\delta n+C_{0}\left(\log n\right)^{p}>0$.
Hence we just need to check the estimate for $\delta$ satisfying
\begin{equation}
C\frac{\left(\log n\right)^{p}}{n}\le\delta<\delta_{0},\label{eq:delta_assumption}
\end{equation}
where $C=C_{0}/c_{0}$ can be made as large as we need by choosing
$C_{0}$ large. Furthermore by choosing $C_{0}$ large enough we can
make sure that the deviations estimate holds trivially for small $n$.
Hence it is enough to check the estimate for $n$ large enough.

Let $l=\left[\delta n\right]$+1, $m=\left[n/l\right]$ and $l'=n-\left(m-1\right)l$.
An application of the Avalanche Principle (using \corref{ldt_delta^2...M^u},
\eqref{delta_assumption}, and \lemref{AppliedAP}) yields
\begin{multline*}
\log\left\Vert \gh_{n}\left(x\right)\right\Vert +\sum_{j=1}^{m-2}\log\left\Vert \gh_{l}\left(x+jl\omega\right)\right\Vert -\log\left\Vert \gh_{l'+l}\left(x+\left(m-2\right)l\omega\right)\right\Vert \\
-\sum_{j=0}^{m-3}\log\left\Vert \gh_{l}\left(x+\left(j+1\right)l\omega\right)\gh_{l}\left(x+jl\omega\right)\right\Vert =O\left(\frac{1}{l}\right)
\end{multline*}
up to a set of measure less than $3n\exp\left(-cl\right)<\exp\left(-c\delta n/2\right)$.
From \eqref{M^a-bound} we can conclude that
\begin{multline*}
\left|\log\left\Vert \gh_{l'}\left(x+\left(m-1\right)l\omega\right)\gh_{l}\left(x+\left(m-2\right)l\omega\right)\right\Vert \right|\\
=\left|\log\left\Vert \gh_{l'+l}\left(x+\left(m-2\right)l\omega\right)\right\Vert \right|\le Cl
\end{multline*}
up to a set of measure less than $\exp\left(-cl\right)\le\exp\left(-c\delta n\right)$.
Hence
\[
\log\left\Vert \gh_{n}\left(x\right)\right\Vert +\sum_{j=1}^{m-2}\log\left\Vert \gh_{l}\left(x+jl\omega\right)\right\Vert -\sum_{j=0}^{m-3}\log\left\Vert \gh_{2l}\left(x+jl\omega\right)\right\Vert =O\left(l\right)
\]
up to a set of measure less than $\exp\left(-c\delta n\right)$. Summing
the above estimate with $x+k\omega$ instead of $x$ yields
\begin{multline*}
\frac{1}{l}\sum_{k=0}^{l-1}\log\left\Vert \gh_{n}\left(x+k\omega\right)\right\Vert +\sum_{j=l}^{\left(m-1\right)l-1}\frac{1}{l}\log\left\Vert \gh_{l}\left(x+j\omega\right)\right\Vert \\
-\sum_{j=0}^{\left(m-2\right)l-1}\frac{1}{l}\log\left\Vert \gh_{2l}\left(x+j\omega\right)\right\Vert =O\left(l\right)
\end{multline*}
up to a set of measure less than $l\exp\left(-c\delta n\right)<\exp\left(-c\delta n/2\right)$.
Using \eqref{M^a-almost-invariance} we can conclude that
\[
\log\left\Vert \gh_{n}\left(x\right)\right\Vert +\sum_{j=l}^{\left(m-1\right)l-1}\frac{1}{l}\log\left\Vert \gh_{l}\left(x+j\omega\right)\right\Vert -\sum_{j=0}^{\left(m-2\right)l-1}\frac{1}{l}\log\left\Vert \gh_{2l}\left(x+j\omega\right)\right\Vert =O\left(l\right)
\]
up to a set of measure less than $\exp\left(-c_{1}\delta n\right)+\exp\left(-c_{2}l\right)<\exp\left(-c\delta n\right)$.
From this, \thmref{sh_ldt}, and \eqref{L^a_l-bounds} it follows
that
\[
\log\left\Vert \gh_{n}\left(x\right)\right\Vert +\left(m-2\right)l\left(L_{l}^{a}-2L_{2l}^{a}\right)=O\left(\delta n\right)
\]
up to a set of measure less than 
\[
2\exp\left(-c_{1}\delta n+C\left(\log n\right)^{p}\right)+\exp\left(-c_{2}\delta n\right)<\exp\left(-c\delta n+C\left(\log n\right)^{p}\right).
\]
Integrating over $\TT$ and using \lemref{Integral-lemma} yields
\[
\left|nL_{n}^{a}+\left(m-2\right)l\left(L_{l}^{a}-2L_{2l}^{a}\right)\right|<C_{1}\delta n+C_{2}n\exp\left(\left(-c\delta n+C\left(\log n\right)^{p}\right)/2\right)<C\delta n.
\]
Note that for the last inequality to hold we need to choose $C$ large
enough in \eqref{delta_assumption}. Now we have that
\[
\left|\log\left\Vert \gh_{n}\left(x\right)\right\Vert -nL_{n}^{a}\right|<C\delta n
\]
up to a set of measure less than $\exp\left(-c\delta n+C\left(\log n\right)^{p}\right)$.
The fact that $L_{n}^{a}$ can be replaced by $L^{a}$ follows from
\lemref{Ln_L} and \eqref{delta_assumption}. \end{proof}
\begin{cor}
\label{cor:ldt_M^u}For any $\delta>0$ and any integer $n>1$ we
have 
\[
\mes\left\{ x\in\TT:\left|\log\left\Vert \fg_{n}\left(x\right)\right\Vert -nL_{n}^{u}\right|>\delta n\right\} <\exp\left(-c_{0}\delta n+C_{0}\left(\log n\right)^{p}\right)
\]
where $c_{0}=c_{0}\left(\left\Vert a\right\Vert _{\infty},\left\Vert b\right\Vert _{*},\left|E\right|,\omega,\gamma\right)$
and $C_{0}=C_{0}\left(\left\Vert a\right\Vert _{\infty},\left\Vert b\right\Vert _{*},\left|E\right|,\omega,\gamma,p\right)$.
The same estimate, with possibly different constants, holds with $L^{u}$
instead of $L_{n}^{u}$.\end{cor}
\begin{proof}
The proof is the same as for \corref{ldt_delta^2...M^u}.
\end{proof}
Next we establish some estimates that will be needed in the next section.
First we prove a uniform upper bound for $\log\left\Vert M_{n}^{a}\right\Vert $.
We will need the following general result about averages of subharmonic
functions.
\begin{lem}
(\cite[Lemma 4.1]{MR2438997}) Let $u$ be a subharmonic function
and let 
\[
u\left(z\right)=\int_{\mathbb{C}}\log\left|z-\zeta\right|d\mu\left(\zeta\right)+h\left(z\right)
\]
be its Riesz representation on a neighborhood of $\cA_{\rho}$. If
$\mu\left(\cA_{\rho}\right)+\left\Vert h\right\Vert _{L^{\infty}\left(\cA_{\rho}\right)}\le M$
then for any $r_{1},r_{2}\in\left(1-\rho,1+\rho\right)$ we have
\[
\left|\left\langle u\left(r_{1}\left(\cdot\right)\right)\right\rangle -\left\langle u\left(r_{2}\left(\cdot\right)\right)\right\rangle \right|\le C_{0}\left|r_{1}-r_{2}\right|,
\]
where $C_{0}=C_{0}\left(M,\rho\right)$.
\end{lem}
The following corollary is an immediate consequence of the previous
lemma and \lemref{M^a-Riesz_bounds}.
\begin{cor}
\label{cor:L(r1)-L(r2)}There exists a constant $C_{0}=C_{0}\left(\left\Vert a\right\Vert _{\infty},\left\Vert b\right\Vert _{*},\left|E\right|,\rho_{0},\rho_{0}',\rho_{0}''\right)$
such that
\[
\left|L_{n}^{u}\left(r_{1}\right)-L_{n}^{u}\left(r_{2}\right)\right|=\left|L_{n}^{a}\left(r_{1}\right)-L_{n}^{a}\left(r_{2}\right)\right|\le C_{0}\left|r_{1}-r_{2}\right|
\]
for any $r_{1},r_{2}\in\left(1-\rho_{0},1+\rho_{0}\right)$ and any
positive integer $n$.\end{cor}
\begin{prop}
\label{prop:M^a-upper-bound} For any integer $n>1$ we have that
\[
\sup_{x\in\TT}\log\left\Vert \gh_{n}\left(x\right)\right\Vert \le nL_{n}^{a}+C_{0}\left(\log n\right)^{p}
\]
where $C_{0}=C_{0}\left(\left\Vert a\right\Vert _{\infty},\left\Vert b\right\Vert _{*},\left|E\right|,\omega,\gamma,p\right)$.\end{prop}
\begin{proof}
It is sufficient to establish the estimate for large $n$. From the
large deviations estimate, with $n\delta=C\left(\log n\right)^{p}$
where $C$ is sufficiently large, we have
\[
\log\left\Vert \gh_{n}\left(rx\right)\right\Vert -nL_{n}^{a}\left(r\right)\le C\left(\log n\right)^{p}
\]
except for a set $\cB\left(r\right)$ of measure less than $\exp\left(-c_{1}C\left(\log n\right)^{p}+C'\left(\log n\right)^{p}\right)<\exp\left(-c\left(\log n\right)^{p}\right)$.
Here $r$ is in a neighborhood of $1$ such that $L^{u}\left(r\right)\ge\gamma/2$.
Such a neighborhood exists because of \corref{L(r1)-L(r2)}. By the
subharmonicity of $\log\left\Vert \gh_{n}\left(z\right)\right\Vert $
we have
\begin{multline}
\log\left\Vert \gh_{n}\left(x\right)\right\Vert -nL_{n}^{a}\le\frac{1}{\pi n^{-2}}\int_{D\left(x,n^{-1}\right)}\left(\log\left\Vert \gh_{n}\left(z\right)\right\Vert -nL_{n}^{a}\right)dA\left(z\right)\\
\le\frac{1}{\pi n^{-2}}\int_{1-n^{-1}}^{1+n^{-1}}\int_{x-2n^{-1}}^{x+2n^{-1}}\left|\log\left\Vert \gh_{n}\left(ry\right)\right\Vert -L_{n}^{a}\right|rdydr.\label{eq:M^a-nL^a...upper_bound}
\end{multline}
For $r\in\left(1-n^{-1},1+n^{-1}\right)$ we have
\begin{multline*}
\int_{x-2n^{-1}}^{x+2n^{-1}}\left|\log\left\Vert \gh_{n}\left(ry\right)\right\Vert -L_{n}^{a}\right|dy\\
\le\int_{x-2n^{-1}}^{x+2n^{-1}}\left|\log\left\Vert \gh_{n}\left(ry\right)\right\Vert -L_{n}^{a}\left(r\right)\right|dy+\left|L_{n}^{a}-L_{n}^{a}\left(r\right)\right|\\
\le C_{1}\left(\log n\right)^{p}n^{-1}+C_{2}n\exp\left(-c\left(\log n\right)^{p}/2\right)+C_{3}n^{-1}<C\left(\log n\right)^{p}n^{-1}.
\end{multline*}
As usual, we used \lemref{Integral-lemma} to deal with the exceptional
set. Plugging this estimate in \eqref{M^a-nL^a...upper_bound} yields
the desired conclusion.
\end{proof}
As was mentioned in the introduction, from this point forward we will
make use of the fact that $\tilde{b}=\bar{b}$ on $\TT$. In particular
we will tacitly use that $D=\tilde{D}$, $S=\tilde{S}$, $L_{n}=L_{n}^{u}$,
$L=L^{u}$, and $\left|\tilde{b}\right|=\left|\bar{b}\right|=\left|b\right|$.

Next we want to estimate $L_{n}\left(E\right)-L_{n}\left(E_{0}\right)$
in a neighborhood of $E_{0}$.
\begin{lem}
\label{lem:M^a(E1)-M^a(E2)}There exist constants $C_{0}=C_{0}\left(\left\Vert a\right\Vert _{\infty},\left\Vert b\right\Vert _{*},\max\left\{ \left|E_{1}\right|,\left|E_{2}\right|\right\} \right)$
and $c_{0}=c_{0}\left(\left\Vert b\right\Vert _{*},\omega\right)$
such that 
\begin{multline*}
\left|\log\left\Vert \fg_{l}\left(x,E_{1}\right)\right\Vert -\log\left\Vert \fg_{l}\left(x,E_{2}\right)\right\Vert \right|\\
=\left|\log\left\Vert \gh_{l}\left(x,E_{1}\right)\right\Vert -\log\left\Vert \gh_{l}\left(x,E_{2}\right)\right\Vert \right|\le\exp\left(C_{0}l\right)\left|E_{1}-E_{2}\right|
\end{multline*}
holds for any positive integer $l$ and any $x$ up to a set (independent
of $E_{1}$ and $E_{2}$) of measure less than $\exp\left(-c_{0}l\right)$.\end{lem}
\begin{proof}
The identity follows from \eqref{logM^u...logM^a}. By the Mean Value
Theorem we have
\begin{multline*}
\left|\log\left\Vert \gh_{l}\left(x,E_{1}\right)\right\Vert -\log\left\Vert \gh_{l}\left(x,E_{2}\right)\right\Vert \right|\\
\le\frac{1}{\min\left\{ \left\Vert \gh_{l}\left(x,E_{1}\right)\right\Vert ,\left\Vert \gh_{l}\left(x,E_{2}\right)\right\Vert \right\} }\left|\left\Vert \gh_{l}\left(x,E_{1}\right)\right\Vert -\left\Vert \gh_{l}\left(x,E_{2}\right)\right\Vert \right|\\
\le\frac{1}{\min\left\{ \left\Vert \gh_{l}\left(x,E_{1}\right)\right\Vert ,\left\Vert \gh_{l}\left(x,E_{2}\right)\right\Vert \right\} }\sup_{E\in\left[E_{1},E_{2}\right]}\left\Vert \frac{\partial}{\partial E}\gh_{l}\left(x,E\right)\right\Vert \left|E_{1}-E_{2}\right|.
\end{multline*}
There exists a constant $C=C\left(\left\Vert a\right\Vert _{\infty},\left\Vert b\right\Vert _{\infty},\max\left\{ \left|E_{1}\right|,\left|E_{2}\right|\right\} \right)$
such that
\[
\sup_{E\in\left[E_{1},E_{2}\right]}\left\Vert \frac{\partial}{\partial E}\gh_{l}\left(x,E\right)\right\Vert \le\exp\left(Cl\right).
\]
The conclusion now follows by using \eqref{M^a-bound}.\end{proof}
\begin{lem}
Fix $E_{0}\in\CC$ such that $L\left(E_{0}\right)\ge\gamma$. There
exist constants $C_{0}=C_{0}\left(\left\Vert a\right\Vert _{\infty},\left\Vert b\right\Vert _{*},\left|E_{0}\right|,\omega,\gamma\right)$,
$C_{1}=C_{1}\left(\left\Vert a\right\Vert _{\infty},\left\Vert b\right\Vert _{*},\left|E_{0}\right|,\omega,\gamma\right)$,
and $n_{0}=n_{0}(\left\Vert a\right\Vert _{\infty},$ $\left\Vert b\right\Vert _{*},\left|E_{0}\right|,\omega,\gamma)$
such that we have
\[
\left|\log\left\Vert \gh_{n}\left(x,E\right)\right\Vert -\log\left\Vert \gh_{n}\left(x,E_{0}\right)\right\Vert \right|\le n^{-C_{0}}
\]
 for $n\ge n_{0}$, $\left|E-E_{0}\right|<n^{-C_{1}}$, and all $x$
up to a set $\cB=\cB\left(n,E_{0}\right)$ of measure less than $n^{-1}$.\end{lem}
\begin{proof}
Let $l=\left[C_{2}\log n\right]$, $m=\left[n/l\right]$, and $l'=n-\left(m-2\right)l$.
In what follows we should keep in mind that some of the estimates
hold by choosing $C_{2}$ large enough. To be able to apply the Avalanche
Principle we will need that $m\ge2$, hence we also need that $n$
is large enough. Applying the Avalanche Principle (see \lemref{AppliedAP})
we get
\begin{multline}
\log\left\Vert \gh_{n}\left(x,E_{0}\right)\right\Vert +\sum_{j=1}^{m-2}\log\left\Vert \gh_{l}\left(x+jl\omega,E_{0}\right)\right\Vert -\log\left\Vert \gh_{l+l'}\left(x+\left(m-2\right)l\omega,E_{0}\right)\right\Vert \\
-\sum_{j=0}^{m-3}\log\left\Vert \gh_{2l}\left(x+jl\omega,E_{0}\right)\right\Vert =O\left(\frac{n}{l}\exp\left(-\frac{\gamma}{2}l\right)\right)=O\left(\frac{1}{n^{cC_{2}}}\right)\label{eq:AP-for-M^a(E0)}
\end{multline}
up to a set of measure $3n\exp\left(-c_{1}l\right)<n^{-cC_{2}}$.
We claim that the Avalanche Principle can be applied, with the same
$\mu$, for the same factorization of $\gh_{n}\left(x,E\right)$.
Note that we cannot apply the deviations estimate since we don't know
whether $L\left(E\right)>0$. For example, \lemref{M^a(E1)-M^a(E2)}
and \corref{ldt_delta^2...M^u} imply that
\begin{multline*}
\log\left\Vert \fg_{l}\left(x,E\right)\right\Vert \ge\log\left\Vert \fg_{l}\left(x,E_{0}\right)\right\Vert -\exp\left(Cl-C_{1}\log n\right)\\
\ge\left(\gamma-\frac{\gamma}{100}\right)l-\exp\left(Cl-C_{1}\log n\right)>\frac{\gamma}{2}l
\end{multline*}
up to a set of measure $\exp\left(-c_{1}l\right)+\exp\left(-c_{2}l\right)<\exp\left(-cl\right)$.
Note that the exceptional set from the deviation estimate is already
included in the exceptional set for \eqref{AP-for-M^a(E0)} and recall
that the exceptional set from \lemref{M^a(E1)-M^a(E2)} doesn't depend
on $E$. Also note that $C_{1}$ needs to satisfy $C_{1}\ge CC_{2}$.
The other estimates needed for the Avalanche Principle are obtained
similarly, provided $C_{1}$ is large enough. Hence, \eqref{AP-for-M^a(E0)}
holds with $E$ instead of $E_{0}$. The conclusion follows by subtracting
\eqref{AP-for-M^a(E0)} for $E$ and $E_{0}$ and using \lemref{M^a(E1)-M^a(E2)}
(again, $C_{1}$ needs to be chosen to be large enough). \end{proof}
\begin{cor}
\label{cor:L(E)-L(E0)} Fix $E_{0}\in\CC$ such that $L\left(E_{0}\right)\ge\gamma$.
There exist constants $C_{0}=C_{0}\left(\left\Vert a\right\Vert _{\infty},\left\Vert b\right\Vert _{*},\left|E_{0}\right|,\omega,\gamma\right)$,
$C_{1}=C_{1}\left(\left\Vert a\right\Vert _{\infty},\left\Vert b\right\Vert _{*},\left|E_{0}\right|,\omega,\gamma\right)$,
and $n_{0}=n_{0}\left(\left\Vert a\right\Vert _{\infty},\left\Vert b\right\Vert _{*},\left|E_{0}\right|,\omega,\gamma\right)$
such that we have
\[
\left|n\left(L_{n}\left(E\right)-L_{n}\left(E_{0}\right)\right)\right|=\left|n\left(L_{n}^{a}\left(E\right)-L_{n}^{a}\left(E_{0}\right)\right)\right|\le n^{-C_{0}}
\]
 for $n\ge n_{0}$ and $\left|E-E_{0}\right|<n^{-C_{1}}$.\end{cor}
\begin{proof}
Integrate the estimate of the previous lemma. To deal with the exceptional
set we used \lemref{Integral-lemma} and the fact that as a consequence
of \eqref{M^a-bound} we have
\[
\left|\log\left\Vert \gh_{n}\left(x,E\right)\right\Vert -\log\left\Vert \gh_{n}\left(x,E_{0}\right)\right\Vert \right|\le\lambda n
\]
up to a set of size $\exp\left(-c\lambda n\right)$ for any $\lambda\ge\lambda_{0}$.
\end{proof}

\section{Estimates for the Entries of the Fundamental Matrix\label{sec:Estimates-for-Entries}}

We will need the following particular case of a lemma from \cite{MR2438997}.
\begin{lem}
\label{lem:sh_better_upper_bound}(\cite[Lemma 2.4]{MR2438997}) Let
$u$ be a subharmonic function defined on $\cA_{\rho}$ such that
$\sup_{\cA_{\rho}}u\le M$. There exist constants $C_{1}=C_{1}\left(\rho\right)$
and $C_{2}$ such that, if for some $0<\delta<1$ and some $L$ we
have 
\[
\mes\left\{ x\in\TT:\, u\left(x\right)<-L\right\} >\delta,
\]
then
\[
\sup_{\TT}u\le C_{1}M-\frac{L}{C_{1}\log\left(C_{2}/\delta\right)}.
\]

\end{lem}

Let $I_{a,E}=\int_{\TT}\log\left|a\left(x\right)-E\right|dx$. Note
that $\left|I_{a,E}\right|<\infty$ if and only if $a\not\equiv E$.
If $a\equiv E$ then it is straightforward to see that $L=0$. Hence
if $L\left(E\right)>0$ then $\left|I_{a,E}\right|<\infty$. Furthermore,
if $L\left(E\right)>0$ on some set, it can be seen that $I_{a,E}$
is continuous in $E$ on that set.
\begin{lem}
\label{lem:considerable-difficulty-lemma}There exists $l_{0}=l_{0}\left(\left\Vert a\right\Vert _{\infty},I_{a,E},\left\Vert b\right\Vert _{*},\left|E\right|,\omega,\gamma\right)$
such that
\[
\mes\left\{ x\in\TT:\,\left|f_{l}\left(x\right)\right|\le\exp\left(-l^{3}\right)\right\} \le\exp\left(-l\right)
\]
for all $l\ge l_{0}$.\end{lem}
\begin{proof}
We argue by contradiction. Assume
\[
\mes\left\{ x\in\TT:\,\left|f_{l}\left(x\right)\right|\le\exp\left(-l^{3}\right)\right\} >\exp\left(-l\right)
\]
for arbitrarily large $l$. We will be tacitly using the fact that
$l$ can be arbitrarily large. We have that
\begin{align*}
\left|\g_{l}\left(x\right)\right| & =\left|f_{l}\left(x\right)\right|\prod_{j=1}^{l}\left|b\left(x+j\omega\right)\right|\le\exp\left(-l^{3}\right)C^{l-1}\le\exp\left(-l^{3}/2\right)
\end{align*}
on a set of measure greater than $\exp\left(-l\right)$. Hence
\[
\mes\left\{ x\in\TT:\,\left|\g_{l}\left(x\right)\right|\le\exp\left(-l^{3}/2\right)\right\} >\exp\left(-l\right).
\]
At the same time we have that
\[
\log\left|\g_{l}\left(x\right)\right|\le\log\left\Vert \gh_{l}\left(x\right)\right\Vert \le Cl
\]
for all $x$, so by applying \lemref{sh_better_upper_bound} we get
that
\[
\left|\g_{l}\left(x\right)\right|\le\exp\left(C_{1}l-\frac{l^{3}}{C_{2}\log\left(C_{3}\exp\left(l\right)\right)}\right)\le\exp\left(-Cl^{2}\right)
\]
for all $x$ and consequently
\begin{equation}
\left|f_{l}\left(x\right)\right|\le\exp\left(\left(l-1\right)(1-D)-C_{1}l^{2}\right)\le\exp\left(-Cl^{2}\right)\label{eq:f_l-ub}
\end{equation}
for all $x$ except for a set of measure less than $\exp\left(-c_{1}\left(l-1\right)+r_{l-1}\right)<\exp\left(-cl\right)$.

From \corref{ldt_M^u} we have that
\begin{multline}
\exp\left(l\gamma\right)\le\left\Vert M_{l}\left(x\right)\right\Vert ^{2}\le2\Bigg(\left|f_{l}\left(x\right)\right|^{2}+\left|f_{l-1}\left(x\right)\right|^{2}\\
+\left|\frac{b\left(x\right)}{b\left(x+\omega\right)}f_{l-1}\left(x+\omega\right)\right|^{2}+\left|\frac{b\left(x\right)}{b\left(x+\omega\right)}f_{l-2}\left(x+\omega\right)\right|^{2}\Bigg)\label{eq:IIMlII_ub}
\end{multline}
for all $x$ except for a set of measure less than $\exp\left(-c_{1}\gamma l/2+r_{l}\right)<\exp\left(-cl\right)$.
Suppose that
\begin{equation}
\left|\frac{b\left(x\right)}{b\left(x+\omega\right)}f_{l-1}\left(x+\omega\right)\right|^{2}\ge\frac{1}{4}\exp\left(l\gamma\right)\label{eq:f_l-1_lb}
\end{equation}
for all $x$ except for a set of measure less than $1/3$ (any constant
in $(0,1/2)$ would work). Since
\begin{multline*}
\frac{\overline{b\left(x\right)}}{b\left(x+l\omega\right)}=\det M_{l}\left(x\right)=-\frac{\overline{b\left(x\right)}}{b\left(x+\omega\right)}f_{l}\left(x\right)f_{l-2}\left(x+\omega\right)\\
+\frac{\overline{b\left(x\right)}}{b\left(x+\omega\right)}f_{l-1}\left(x\right)f_{l-1}\left(x+\omega\right)
\end{multline*}
it follows that 
\begin{multline*}
\left|f_{l-1}\left(x\right)\right|=\left|\frac{b\left(x\right)}{b\left(x+\omega\right)}f_{l-1}\left(x+\omega\right)\right|^{-1}\left|\frac{b\left(x\right)}{b\left(x+l\omega\right)}+\frac{b\left(x\right)}{b\left(x+\omega\right)}f_{l}\left(x\right)f_{l-2}\left(x+\omega\right)\right|\\
\le2\exp\left(-l\gamma/2\right)\left(C_{1}\exp\left(\delta-D\right)+\exp\left(-C_{2}l^{2}+C_{3}l\right)\right)
\end{multline*}
for all $x$ except for a set of measure less than $1/3+\exp\left(-c_{1}\delta+r_{1}\right)+\exp\left(-c_{1}l\right)+\exp\left(-c_{2}l+r_{l}^{'}\right)$.
Note that in the above estimate we used
\[
\log\left|\frac{b\left(x\right)}{b\left(x+\omega\right)}f_{l-2}\left(x+\omega\right)\right|\le\log\left\Vert M_{l}\left(x\right)\right\Vert 
\]
and the large deviations estimate for $M_{l}$. Choosing $\delta=l\gamma/2$
we get
\[
\left|f_{l-1}\left(x\right)\right|\le C
\]
for all $x$ except for a set of measure less than $1/3+\exp\left(-cl\right)$.
This contradicts \eqref{f_l-1_lb} because 
\[
\left|f_{l-1}\left(x+\omega\right)\right|\le C
\]
and \eqref{f_l-1_lb} would hold at the same time on a set of measure
greater than $1/3-\exp\left(-cl\right)$. Hence we must have
\begin{equation}
\left|\frac{b\left(x\right)}{b\left(x+\omega\right)}f_{l-1}\left(x+\omega\right)\right|^{2}<\frac{1}{4}\exp\left(l\gamma\right)\label{eq:f_l-1-ub}
\end{equation}
on a set of measure greater than $1/3$. At the same time
\begin{multline*}
\exp\left(\left(l+1\right)\gamma\right)\le\left\Vert M_{l+1}\left(x\right)\right\Vert ^{2}\\
\le2\left(\left|f_{l+1}\left(x\right)\right|^{2}+\left|f_{l}\left(x\right)\right|^{2}+\left|\frac{b\left(x\right)}{b\left(x+\omega\right)}f_{l}\left(x+\omega\right)\right|^{2}+\left|\frac{b\left(x\right)}{b\left(x+\omega\right)}f_{l-1}\left(x+\omega\right)\right|^{2}\right)
\end{multline*}
for all $x$ except for a set of measure less than 
\[
\exp\left(-c_{1}\gamma\left(l+1\right)/2+r_{l+1}\right)<\exp\left(-cl\right).
\]
This, \eqref{f_l-ub}, and \eqref{f_l-1-ub} imply that we must have
\begin{multline*}
\left|f_{l+1}\left(x\right)\right|^{2}\ge\frac{1}{2}\exp\left(\left(l+1\right)\gamma\right)-\exp\left(-C_{1}l^{2}\right)\\
-C_{2}\exp\left(l-D-C_{1}l^{2}\right)-\frac{1}{4}\exp\left(l\gamma\right)>\frac{1}{4}\exp\left(l\gamma\right)
\end{multline*}
on a set of measure greater than 
\[
\frac{1}{3}-\exp\left(-c_{1}l\right)-2\exp\left(-c_{2}l\right)-\exp\left(-c_{3}l+r_{1}\right)>\frac{1}{3}-\exp\left(-cl\right).
\]
From 
\begin{multline*}
\frac{\overline{b\left(x\right)}}{b\left(x+\left(l+1\right)\omega\right)}=\det M_{l+1}\left(x\right)=-\frac{\overline{b\left(x\right)}}{b\left(x+\omega\right)}f_{l+1}\left(x\right)f_{l-1}\left(x+\omega\right)\\
+\frac{\overline{b\left(x\right)}}{b\left(x+\omega\right)}f_{l}\left(x\right)f_{l}\left(x+\omega\right)
\end{multline*}
it can be seen that
\begin{multline*}
\left|\frac{b\left(x\right)}{b\left(x+\omega\right)}f_{l-1}\left(x+\omega\right)\right|=\left|f_{l+1}\left(x\right)\right|^{-1}\\
\cdot\left|\frac{b\left(x\right)}{b\left(x+\left(l+1\right)\omega\right)}-\frac{b\left(x\right)}{b\left(x+\omega\right)}f_{l}\left(x\right)f_{l}\left(x+\omega\right)\right|\\
\le2\exp\left(-l\gamma/2\right)\left(C_{1}\exp\left(\delta-D\right)+C_{1}\exp\left(\delta-D-C_{2}l^{2}\right)\right)
\end{multline*}
on a set of measure greater than $1/3-\exp\left(-c_{1}l\right)-2\exp\left(-c_{2}\delta+r_{1}\right)-2\exp\left(-c_{3}l\right)$.
Choosing $\delta=l\gamma/5$ we get
\begin{equation}
\left|\frac{b\left(x\right)}{b\left(x+\omega\right)}f_{l-1}\left(x+\omega\right)\right|\le\exp\left(-l\gamma/4\right)\label{eq:b/bfl-1_ub}
\end{equation}
on a set of measure greater than $1/3-\exp\left(-cl\right)$. We will
contradict \eqref{IIMlII_ub} by showing that
\begin{equation}
\left|f_{l}\left(x\right)\right|^{2}+\left|f_{l-1}\left(x\right)\right|^{2}+\left|\frac{b\left(x\right)}{b\left(x+\omega\right)}f_{l-1}\left(x+\omega\right)\right|^{2}+\left|\frac{b\left(x\right)}{b\left(x+\omega\right)}f_{l-2}\left(x+\omega\right)\right|^{2}\le C\label{eq:fl+fl-1+fl-1+fl-2<=00003DC}
\end{equation}
on a set of measure greater than $1/3-\exp\left(-cl\right)$. Let
$G_{l}$ be the set on which \eqref{b/bfl-1_ub} holds.

By writing
\[
M_{l}\left(x+\omega\right)=\frac{1}{b\left(x+\left(l+1\right)\omega\right)}\left[\begin{array}{cc}
a\left(x+l\omega\right)-E & -\overline{b\left(x+l\omega\right)}\\
b\left(x+\left(l+1\right)\omega\right) & 0
\end{array}\right]M_{l-1}\left(x+\omega\right)
\]
we get
\[
f_{l}\left(x+\omega\right)=\frac{a\left(x+l\omega\right)-E}{b\left(x+\left(l+1\right)\omega\right)}f_{l-1}\left(x+\omega\right)-\frac{\overline{b\left(x+l\omega\right)}}{b\left(x+\left(l+1\right)\omega\right)}f_{l-2}\left(x+\omega\right).
\]
 From this we deduce that
\begin{multline*}
\left|\frac{b\left(x\right)}{b\left(x+\omega\right)}f_{l-2}\left(x+\omega\right)\right|=\left|\frac{b\left(x+\left(l+1\right)\omega\right)}{b\left(x+l\omega\right)}\right|\\
\cdot\left|\frac{a\left(x+l\omega\right)-E}{b\left(x+\left(l+1\right)\omega\right)}\frac{b\left(x\right)}{b\left(x+\omega\right)}f_{l-1}\left(x+\omega\right)-\frac{b\left(x\right)}{b\left(x+\omega\right)}f_{l}\left(x+\omega\right)\right|\\
\le C_{1}\exp\left(\delta-D\right)\left(C_{2}\exp\left(\delta-D-\gamma l/4\right)+C_{1}\exp\left(\delta-D-C_{3}l^{2}\right)\right)
\end{multline*}
on a subset of $G_{l}$ of measure greater than 
\[
\frac{1}{3}-3\exp\left(-c_{1}\delta+r_{1}\right)-\exp\left(-c_{2}l\right)-\exp\left(-c_{3}l\right).
\]
By choosing $\delta=\gamma l/17$ we get
\[
\left|\frac{b\left(x\right)}{b\left(x+\omega\right)}f_{l-2}\left(x+\omega\right)\right|\le\exp\left(-\gamma l/8\right)
\]
on a subset of $G_{l}$ of measure greater than $1/3-\exp\left(-cl\right)$.

By writing 
\[
M_{l}\left(x-\omega\right)=M_{l-1}\left(x\right)\frac{1}{b\left(x\right)}\left[\begin{array}{cc}
a\left(x-\omega\right)-E & -\overline{b\left(x-\omega\right)}\\
b\left(x\right) & 0
\end{array}\right]
\]
we get
\[
f_{l}\left(x-\omega\right)=\frac{a\left(x-\omega\right)-E}{b\left(x\right)}f_{l-1}\left(x\right)-\frac{\overline{b\left(x\right)}}{b\left(x+\omega\right)}f_{l-2}\left(x+\omega\right).
\]
From this we deduce that
\begin{multline*}
\left|f_{l-1}\left(x\right)\right|=\left|\frac{a\left(x-\omega\right)-E}{b\left(x\right)}\right|^{-1}\left|f_{l}\left(x-\omega\right)+\frac{\overline{b\left(x\right)}}{b\left(x+\omega\right)}f_{l-2}\left(x+\omega\right)\right|\\
\le C_{1}\exp\left(\delta-I_{a,E}\right)\left(\exp\left(-C_{1}l^{2}\right)+\exp\left(-\gamma l/8\right)\right)
\end{multline*}
on a subset of $G_{l}$ of measure greater than $1/3-\exp\left(-c_{1}\delta+r_{1}\right)-\exp\left(-c_{3}l\right)-\exp\left(-c_{4}l\right)$.
By choosing $\delta=\gamma l/17$ we get
\[
\left|f_{l-1}\left(x\right)\right|\le\exp\left(-\gamma l/16\right)
\]
on a subset of $G_{l}$ of measure greater than $1/3-\exp\left(-cl\right)$.$ $
Now it is easy to see that we have \eqref{fl+fl-1+fl-1+fl-2<=00003DC}. \end{proof}
\begin{lem}
\label{lem:small_fl_is_small}Let $\sigma>0$. There exist constants
$l_{0}=l_{0}(\left\Vert a\right\Vert _{\infty},I_{a,E},\left\Vert b\right\Vert _{*},\left|E\right|,\omega,\gamma,$
$\sigma)$ and $N_{0}=N_{0}(\left\Vert a\right\Vert _{\infty},I_{a,E},$
$\left\Vert b\right\Vert _{*},\left|E\right|,\omega,\gamma,\sigma)$
such that
\[
\mes\left\{ x\in\TT:\,\left|f_{l}\left(x\right)\right|\le\exp\left(-N^{\sigma}\right)\right\} \le\exp\left(-N^{\sigma}l^{-2}\right)
\]
for any $N\ge N_{0}$ and for any $l_{0}\le l\le N^{\sigma/3}$. The
same result, but with possibly different $l_{0}$ and $N_{0}$, holds
for $\f_{l}$.\end{lem}
\begin{proof}
We argue by contradiction. Assume
\[
\mes\left\{ x\in\TT:\,\left|f_{l}\left(x\right)\right|\le\exp\left(-N^{\sigma}\right)\right\} >\exp\left(-N^{\sigma}l^{-2}\right)
\]
 for some arbitrarily large $l$ and $N$. We have that
\begin{align*}
\left|\g_{l}\left(x\right)\right| & =\left|f_{l}\left(x\right)\right|\prod_{j=1}^{l}\left|b\left(x+j\omega\right)\right|\le\exp\left(-N^{\sigma}\right)C^{l-1}\le\exp\left(-N^{\sigma}/2\right)
\end{align*}
on a set of measure greater than $\exp\left(-N^{\sigma}l^{-2}\right)$.
Hence
\[
\mes\left\{ x\in\TT:\,\left|\g_{l}\left(x\right)\right|\le\exp\left(-N^{\sigma}/2\right)\right\} >\exp\left(-N^{\sigma}l^{-2}\right).
\]
By applying \lemref{sh_better_upper_bound} we get that
\[
\left|\g_{l}\left(x\right)\right|\le\exp\left(C_{1}l-\frac{N^{\sigma}}{2C_{1}\log\left(C_{2}\exp\left(N^{\sigma}l^{-2}\right)\right)}\right)\le\exp\left(-Cl^{2}\right)
\]
for all $x$. Note that the last inequality is equivalent to 
\[
\frac{C_{1}}{l}+C\le\frac{N^{\sigma}l^{-2}}{2C_{1}\log\left(C_{2}\exp\left(N^{\sigma}l^{-2}\right)\right)}=\frac{N^{\sigma}l^{-2}}{2C_{1}\log C_{2}+2C_{1}N^{\sigma}l^{-2}}
\]
which clearly holds with $C=1/\left(4C_{1}\right)$ for large $l$
and $N,$ since $N^{\sigma}l^{-2}\ge N^{\sigma/3}$. We now have that
\[
\left|f_{l}\left(x\right)\right|\le\exp\left(\left(l-1\right)(1-D)-C'l^{2}\right)\le\exp\left(-Cl^{2}\right)
\]
for all $x$ except for a set of measure less than $\exp\left(-c_{1}\left(l-1\right)+r_{l-1}\right)<\exp\left(-cl\right)$.
The contradiction follows in the same way as in the previous lemma.

To get the result for $\f_{l}$ one can argue by contradiction. Using
\[
\left|\g_{l}\left(x\right)\right|=\left|\f_{l}\left(x\right)\right|\prod_{j=0}^{n-1}\left|b\left(x+j\omega\right)b\left(x+\left(j+1\right)\omega\right)\right|^{1/2}
\]
one can get that $\left|\g_{l}\left(x\right)\right|\le\exp\left(-Cl^{2}\right)$
for all $x$ and this gives the same contradiction as before.
\end{proof}
We recall for convenience some facts about stability of contracting
and expanding directions of unimodular matrices. It follows from the
polar decomposition that if $A\in SL\left(2,\CC\right)$ then there
exist unit vectors $u_{A}^{+}\perp u_{A}^{-}$ and $v_{A}^{+}\perp v_{A}^{-}$
such that $Au_{A}^{+}=\left\Vert A\right\Vert v_{A}^{+}$ and $Au_{A}^{-}=\left\Vert A\right\Vert ^{-1}v_{A}^{-}$.
\begin{lem}
\label{lem:expanding_contracting_directions}(\cite[Lemma 2.5]{MR2438997})
For any $A$, $B\in SL\left(2,\CC\right)$ we have
\begin{align*}
\left|Bu_{AB}^{-}\wedge u_{A}^{-}\right| & \le\left\Vert A\right\Vert ^{-2}\left\Vert B\right\Vert ,\,\left|u_{BA}^{-}\wedge u_{A}^{-}\right|\le\left\Vert A\right\Vert ^{-2}\left\Vert B\right\Vert ^{2}\\
\left|v_{AB}^{+}\wedge v_{A}^{+}\right| & \le\left\Vert A\right\Vert ^{-2}\left\Vert B\right\Vert ^{2},\,\left|v_{BA}^{+}\wedge Bv_{A}^{+}\right|\le\left\Vert A\right\Vert ^{-2}\left\Vert B\right\Vert .
\end{align*}

\end{lem}
We will need the following estimate (cf. \cite[(2.35)]{MR2438997})
in the proof of \lemref{three_determinants}.
\begin{lem}
\label{lem:wedge_triangle_ineq} If $A\in SL\left(2,\CC\right)$ and
$w_{1}$, $w_{2}$, and $w_{3}$ are unit vectors in the plane then
\[
\left|w_{1}\wedge Aw_{2}\right|\le\left|w_{1}\wedge Aw_{3}\right|+\sqrt{2}\left\Vert A^{-1}\right\Vert \left|w_{2}\wedge w_{3}\right|
\]
and
\[
\left|w_{1}\wedge Aw_{2}\right|\le\left|w_{3}\wedge Aw_{2}\right|+\sqrt{2}\left\Vert A\right\Vert \left|w_{1}\wedge w_{3}\right|
\]
\end{lem}
\begin{proof}
Since $A$ preserves area we have
\begin{multline*}
\left|w_{1}\wedge Aw_{2}\right|=\left|A^{-1}w_{1}\wedge w_{2}\right|\le\left|A^{-1}w_{1}\wedge w_{3}\right|+\min\left|A^{-1}w_{1}\wedge\left(w_{2}\pm w_{3}\right)\right|\\
\le\left|w_{1}\wedge Aw_{3}\right|+\left\Vert A^{-1}w_{1}\right\Vert \min\left\Vert w_{2}\pm w_{3}\right\Vert \le\left|w_{1}\wedge Aw_{3}\right|+\left\Vert A^{-1}\right\Vert \sqrt{2}\left|w_{2}\wedge w_{3}\right|.
\end{multline*}
The second inequality follows from the first one.
\end{proof}
Let $\cG_{N}$ be the set of points $x\in\TT$ such that for any $1\le j\le N$
and $\left|l\right|\le2N$ we have $\left|\log\left\Vert \fg_{j}\left(x+l\omega\right)\right\Vert -jL\right|\le N^{\sigma}$,
$\log\left\Vert \fg_{j}\left(x+l\omega\right)^{-1}\right\Vert \le N^{\sigma}$,
and $\left|\log\left|b\left(x+j\omega\right)\right|-D\right|\le N^{\sigma}$.
From \corref{ldt_M^u}, \corref{M^u^-1bound} and \thmref{sh_ldt}
we have that 
\begin{multline*}
\mes\left(\TT\setminus\cG_{N}\right)\le\left(4N+1\right)N\exp\left(-c_{1}N^{\sigma}+r_{N}\right)+\left(4N+1\right)N\exp\left(-c_{2}N^{\sigma}\right)\\
+N\exp\left(-c_{3}N^{\sigma}+r'_{1}\right)\le\exp\left(-cN^{\sigma}\right)
\end{multline*}
 for $N$ large enough. The choice of $\mathcal{G}_{N}$ is such that
all the estimates in the next lemma hold on this set.
\begin{lem}
\label{lem:three_determinants} Let $0<\sigma<1$. There exist constants
$l_{0}=l_{0}(\left\Vert a\right\Vert _{\infty},I_{a,E},\left\Vert b\right\Vert _{*},\left|E\right|,\omega,$
$\gamma,\sigma)$ and $N_{0}=N_{0}(\left\Vert a\right\Vert _{\infty},I_{a,E},\left\Vert b\right\Vert _{*},\left|E\right|,\omega,\gamma,\sigma)$
such that
\begin{multline}
\mes\left\{ x\in\TT:\,\left|\f_{N}\left(x\right)\right|+\left|\f_{N}\left(x+j_{1}\omega\right)\right|+\left|\f_{N}\left(x+j_{2}\omega\right)\right|\le\exp\left(NL_{N}-100N^{\sigma}\right)\right\} \\
\le\exp\left(-N^{\sigma/2}\right)\label{eq:three_determinants_estimate}
\end{multline}
for any $l_{0}\le j_{1}\le j_{1}+l_{0}\le j_{2}\le N^{\sigma/8}$
and $N\ge N_{0}$.\end{lem}
\begin{proof}
Let $\left\{ e_{1},e_{2}\right\} $ be the standard basis of $\mathbb{R}^{2}$.
By \eqref{M^u-entries} we have 
\begin{multline*}
\f_{N}\left(x\right)=\fg_{N}\left(x\right)e_{1}\wedge e_{2}\\
=\left(\fg_{N}\left(x\right)\left[\left(u_{N}^{+}\left(x\right)\cdot e_{1}\right)u_{N}^{+}\left(x\right)+\left(u_{N}^{-}\left(x\right)\cdot e_{1}\right)u_{N}^{-}\left(x\right)\right]\right)\wedge e_{2}\\
=\left(u_{N}^{+}\left(x\right)\cdot e_{1}\right)\left\Vert \fg_{N}\left(x\right)\right\Vert v_{N}^{+}\left(x\right)\wedge e_{2}+\left(u_{N}^{-}\left(x\right)\cdot e_{1}\right)\left\Vert \fg_{N}\left(x\right)\right\Vert ^{-1}v_{N}^{-}\left(x\right)\wedge e_{2}.
\end{multline*}
If $\left|\f_{N}\left(x\right)\right|\le\exp\left(NL_{N}-100N^{\sigma}\right)$
then
\begin{multline*}
\left\Vert \fg_{N}\left(x\right)\right\Vert \left|u_{N}^{+}\left(x\right)\cdot e_{1}\right|\left|v_{N}^{+}\left(x\right)\wedge e_{2}\right|-\left\Vert \fg_{N}\left(x\right)\right\Vert ^{-1}\left|u_{N}^{-}\left(x\right)\cdot e_{1}\right|\left|v_{N}^{-}\left(x\right)\wedge e_{2}\right|\\
\le\exp\left(NL_{N}-100N^{\sigma}\right).
\end{multline*}
From the above and the fact that $u_{N}^{+}\left(x\right)\cdot e_{1}=u_{N}^{-}\left(x\right)\wedge e_{1}$
(recall that $u_{N}^{+}\perp u_{N}^{-}$) one gets that on $\cG_{N}$
we have
\begin{multline*}
\left|u_{N}^{-}\left(x\right)\wedge e_{1}\right|\left|v_{N}^{+}\left(x\right)\wedge e_{2}\right|\le\exp\left(N\left(L_{N}-L\right)-99N^{\sigma}\right)+\exp\left(2N^{\sigma}-2NL\right)\\
\le\exp\left(-90N^{\sigma}\right)
\end{multline*}
and hence $\left|u_{N}^{-}\left(x\right)\wedge e_{1}\right|\le\exp\left(-40N^{\sigma}\right)$
or $\left|v_{N}^{+}\left(x\right)\wedge e_{2}\right|\le\exp\left(-40N^{\sigma}\right)$.

Suppose \eqref{three_determinants_estimate} fails. Then
\begin{multline*}
\mes\left\{ x\in\cG_{N}:\,\left|\f_{N}\left(x\right)\right|+\left|\f_{N}\left(x+j_{1}\omega\right)\right|+\left|\f_{N}\left(x+j_{2}\omega\right)\right|\le\exp\left(NL_{N}-100N^{\sigma}\right)\right\} \\
>\exp\left(-N^{\sigma/2}\right)-\exp\left(-c_{1}N^{\sigma}\right)>\exp\left(-cN^{\sigma/2}\right).
\end{multline*}
Let $x$ be in the above set. By the preliminary discussion, either
$\left|u_{N}^{-}\left(x\right)\wedge e_{1}\right|\le\exp\left(-40N^{\sigma}\right)$
or $\left|v_{N}^{+}\left(x\right)\wedge e_{2}\right|\le\exp\left(-40N^{\sigma}\right)$
has to hold for two of the points $x$, $x+j_{1}\omega$, $x+j_{2}\omega$.

We first assume that 
\begin{equation}
\left|u_{N}^{-}\left(x+j_{1}\omega\right)\wedge e_{1}\right|\le\exp\left(-40N^{\sigma}\right)\quad\text{and}\quad\left|u_{N}^{-}\left(x+j_{2}\omega\right)\wedge e_{1}\right|\le\exp\left(-40N^{\sigma}\right).\label{eq:first_catch}
\end{equation}
 We now compare $\fg_{j_{2}-j_{1}}\left(x+j_{1}\omega\right)u_{N}^{-}\left(x+j_{1}\omega\right)$
and $u_{N}^{-}\left(x+j_{2}\omega\right)$. From \lemref{wedge_triangle_ineq}
it follows that 
\begin{multline*}
\left|u_{N}^{-}\left(x+j_{2}\omega\right)\wedge\fg_{j_{2}-j_{1}}\left(x+j_{1}\omega\right)u_{N}^{-}\left(x+j_{1}\omega\right)\right|\\
\le\left|u_{N}^{-}\left(x+j_{2}\omega\right)\wedge\fg_{j_{2}-j_{1}}\left(x+j_{1}\omega\right)u_{N+j_{2}-j_{1}}^{-}\left(x+j_{1}\omega\right)\right|\\
+C\left\Vert \fg_{j_{2}-j_{1}}\left(x+j_{1}\omega\right)^{-1}\right\Vert \left|u_{N+j_{2}-j_{1}}^{-}\left(x+j_{1}\omega\right)\wedge u_{N}^{-}\left(x+j_{1}\omega\right)\right|
\end{multline*}
Applying \lemref{expanding_contracting_directions} with $A=\fg_{N}\left(x+j_{2}\omega\right)$
and $B=\fg_{j_{2}-j_{1}}\left(x+j_{1}\omega\right)$ for the first
term, and $A=\fg_{N}\left(x+j_{1}\omega\right)$ and $B=\fg_{j_{2}-j_{1}}\left(x+\left(N+j_{1}\right)\omega\right)$
for the second term, yields
\begin{multline}
\left|u_{N}^{-}\left(x+j_{2}\omega\right)\wedge\fg_{j_{2}-j_{1}}\left(x+j_{1}\omega\right)u_{N}^{-}\left(x+j_{1}\omega\right)\right|\\
\le\left\Vert \fg_{N}\left(x+j_{2}\omega\right)\right\Vert ^{-2}\left\Vert \fg_{j_{2}-j_{1}}\left(x+j_{1}\omega\right)\right\Vert \\
+C\left\Vert \fg_{j_{2}-j_{1}}\left(x+j_{1}\omega\right)^{-1}\right\Vert \left\Vert \fg_{N}\left(x+j_{1}\omega\right)\right\Vert ^{-2}\left\Vert \fg_{j_{2}-j_{1}}\left(x+\left(N+j_{1}\right)\omega\right)\right\Vert ^{2}\\
\le\exp\left(\left(-2N+j_{2}-j_{1}\right)L+3N^{\sigma}\right)+C\exp\left(\left(-2N+2\left(j_{2}-j_{1}\right)\right)L+5N^{\sigma}\right)\\
\le\exp\left(-NL\right)\label{eq:uN_vs_tMuN}
\end{multline}
for $x\in\cG_{N}$. Using \lemref{wedge_triangle_ineq}, \eqref{first_catch},
and \eqref{uN_vs_tMuN} we get
\begin{multline*}
\left|e_{1}\wedge\fg_{j_{2}-j_{1}}\left(x+j_{1}\omega\right)e_{1}\right|\le\left|e_{1}\wedge\fg_{j_{2}-j_{1}}\left(x+j_{1}\omega\right)u_{N}^{-}\left(x+j_{1}\omega\right)\right|\\
+C\left\Vert \fg_{j_{2}-j_{1}}\left(x+j_{1}\omega\right)^{-1}\right\Vert \left|e_{1}\wedge u_{N}^{-}\left(x+j_{1}\omega\right)\right|\\
\le\left|u_{N}^{-}\left(x+j_{2}\omega\right)\wedge\fg_{j_{2}-j_{1}}\left(x+j_{1}\omega\right)u_{N}^{-}\left(x+j_{1}\omega\right)\right|\\
+C\left\Vert \fg_{j_{2}-j_{1}}\left(x+j_{1}\omega\right)\right\Vert \left|e_{1}\wedge u_{N}^{-}\left(x+j_{2}\omega\right)\right|\\
+C\left\Vert \fg_{j_{2}-j_{1}}\left(x+j_{1}\omega\right)^{-1}\right\Vert \left|e_{1}\wedge u_{N}^{-}\left(x+j_{1}\omega\right)\right|\\
\le\exp\left(-NL\right)+C\exp\left(\left(j_{2}-j_{1}\right)L-39N^{\sigma}\right)+C\exp\left(-39N^{\sigma}\right)\\
\le\exp\left(-30N^{\sigma}\right).
\end{multline*}
On the other hand by \eqref{M^u-entries} we have 
\[
\left|e_{1}\wedge\fg_{j_{2}-j_{1}}\left(x+j_{1}\omega\right)e_{1}\right|=\left|\frac{b\left(x+j_{2}\omega\right)}{b\left(x+\left(j_{2}-1\right)\omega\right)}\right|^{1/2}\left|\f_{j_{2}-j_{1}-1}\left(x+j_{1}\omega\right)\right|,
\]
so
\[
\left|\f_{j_{2}-j_{1}-1}\left(x+j_{1}\omega\right)\right|\le C\exp\left(\frac{1}{2}\left(N^{\sigma}-D\right)-30N^{\sigma}\right)\le\exp\left(-20N^{\sigma}\right).
\]
The same type of estimate is obtained if we replace $\left(j_{1},j_{2}\right)$
in \eqref{first_catch} with $\left(0,j_{1}\right)$ or $\left(0,j_{2}\right)$.

Now assume that
\[
\left|v_{N}^{+}\left(x+j_{1}\omega\right)\wedge e_{2}\right|\le\exp\left(-40N^{\sigma}\right)\quad\text{and}\quad\left|v_{N}^{+}\left(x+j_{2}\omega\right)\wedge e_{2}\right|\le\exp\left(-40N^{\sigma}\right).
\]
Similarly to the previous case (first use \lemref{wedge_triangle_ineq}
and then \lemref{expanding_contracting_directions}) we have
\begin{multline*}
\left|v_{N}^{+}\left(x+j_{2}\omega\right)\wedge\fg_{j_{2}-j_{1}}\left(x+\left(N+j_{1}\right)\omega\right)v_{N}^{+}\left(x+j_{1}\omega\right)\right|\\
\le\left|v_{N+j_{2}-j_{1}}^{+}\left(x+j_{1}\omega\right)\wedge\fg_{j_{2}-j_{1}}\left(x+\left(N+j_{1}\right)\omega\right)v_{N}^{+}\left(x+j_{1}\omega\right)\right|\\
+C\left\Vert \fg_{j_{2}-j_{1}}\left(x+\left(N+j_{1}\right)\omega\right)\right\Vert \left|v_{N}^{+}\left(x+j_{2}\omega\right)\wedge v_{N+j_{2}-j_{1}}^{+}\left(x+j_{1}\omega\right)\right|\\
\le\left\Vert \fg_{N}\left(x+j_{1}\omega\right)\right\Vert ^{-2}\left\Vert \fg_{j_{2}-j_{1}}\left(x+\left(N+j_{1}\right)\omega\right)\right\Vert \\
+C\left\Vert \fg_{j_{2}-j_{1}}\left(x+\left(N+j_{1}\right)\omega\right)\right\Vert \left\Vert \fg_{N}\left(x+j_{2}\omega\right)\right\Vert ^{-2}\left\Vert \fg_{j_{2}-j_{1}}\left(x+j_{1}\omega\right)\right\Vert ^{2}\\
\le\exp\left(\left(-2N+j_{2}-j_{1}\right)L+3N^{\sigma}\right)+C\exp\left(\left(-2N+3\left(j_{2}-j_{1}\right)\right)L+5N^{\sigma}\right)\\
\le\exp\left(-NL\right)
\end{multline*}
for $x\in\cG_{N}$ and 
\begin{multline*}
\left|e_{2}\wedge\fg_{j_{2}-j_{1}}\left(x+\left(N+j_{1}\right)\omega\right)e_{2}\right|\le\left|e_{2}\wedge\fg_{j_{2}-j_{1}}\left(x+\left(N+j_{1}\right)\omega\right)v_{N}^{+}\left(x+j_{1}\omega\right)\right|\\
+C\left\Vert \fg_{j_{2}-j_{1}}\left(x+\left(N+j_{1}\right)\omega\right)^{-1}\right\Vert \left|e_{2}\wedge v_{N}^{+}\left(x+j_{1}\omega\right)\right|\\
\le\left|v_{N}^{+}\left(x+j_{2}\omega\right)\wedge\fg_{j_{2}-j_{1}}\left(x+\left(N+j_{1}\right)\omega\right)v_{N}^{+}\left(x+j_{1}\omega\right)\right|\\
+C\left\Vert \fg_{j_{2}-j_{1}}\left(x+\left(N+j_{1}\right)\omega\right)\right\Vert \left|e_{2}\wedge v_{N}^{+}\left(x+j_{2}\omega\right)\right|\\
+C\left\Vert \fg_{j_{2}-j_{1}}\left(x+\left(N+j_{1}\right)\omega\right)^{-1}\right\Vert \left|e_{2}\wedge v_{N}^{+}\left(x+j_{1}\omega\right)\right|
\end{multline*}

\begin{multline*}
\le\exp\left(-NL\right)+C\exp\left(\left(j_{2}-j_{1}\right)L-39N^{\sigma}\right)+C\exp\left(-39N^{\sigma}\right)\le\exp\left(-30N^{\sigma}\right).
\end{multline*}
On the other hand by \eqref{M^u-entries} we have
\begin{multline*}
\left|e_{2}\wedge\fg_{j_{2}-j_{1}}\left(x+\left(N+j_{1}\right)\omega\right)e_{2}\right|=\left|\frac{b\left(x+\left(N+j_{1}\right)\omega\right)}{b\left(x+\left(N+j_{1}+1\right)\omega\right)}\right|^{1/2}\\
\cdot\left|\f_{j_{2}-j_{1}-1}\left(x+\left(N+j_{1}+1\right)\omega\right)\right|,
\end{multline*}
so 
\[
\left|\f_{j_{2}-j_{1}-1}\left(x+\left(N+j_{1}+1\right)\omega\right)\right|\le C\exp\left(\frac{1}{2}\left(N^{\sigma}-D\right)-30N^{\sigma}\right)\le\exp\left(-20N^{\sigma}\right).
\]

In conclusion
\[
\mes\left\{ x\in\TT:\,\left|\f_{l}\left(x\right)\right|\le\exp\left(-20N^{\sigma}\right)\right\} >\exp\left(-cN^{\sigma/2}\right)
\]
for some choice of $l$ from $j_{1}-1$, $j_{2}-1$, $j_{2}-j_{1}-1$.
However, this contradicts the fact that \lemref{small_fl_is_small}
implies
\begin{multline*}
\mes\left\{ x\in\TT:\,\left|\f_{l}\left(x\right)\right|\le\exp\left(-20N^{\sigma}\right)\right\} \le\mes\left\{ x\in\TT:\,\left|\f_{l}\left(x\right)\right|\le\exp\left(-N^{\sigma}\right)\right\} \\
\le\exp\left(-N^{\sigma}l^{-2}\right)\le\exp\left(-N^{3\sigma/4}\right)<\exp\left(-cN^{\sigma/2}\right)
\end{multline*}
(we used $l\le N^{\sigma/8}$).\end{proof}
\begin{lem}
\label{lem:avg_entries->lower_bound} There exist constants $\kappa>0$
and $N_{0}=N_{0}\left(\left\Vert a\right\Vert _{\infty},I_{a,E},\left\Vert b\right\Vert _{*},\left|E\right|,\omega,\gamma\right)$
such that for $N\ge N_{0}$ we have 
\[
\int_{\TT}\frac{1}{N}\left|\f_{N}\left(x\right)\right|dx>L_{N}-N^{-\kappa}.
\]
\end{lem}
\begin{proof}
Let $\Omega_{N}$ be the set of points $x\in\cG_{N}$ such that
\begin{multline*}
\min\{\left|\f_{N}\left(x+j_{1}\omega\right)\right|+\left|\f_{N}\left(x+j_{2}\omega\right)\right|+\left|\f_{N}\left(x+j_{3}\omega\right)\right|\\
:\,0<j_{1}<j_{1}+l_{0}\le j_{2}<j_{2}+l_{0}\le j_{3}\le N^{\sigma/8}\}\\
>\exp\left(NL_{N}-100N^{\sigma}\right),
\end{multline*}
where $l_{0}$ is as in the previous lemma. If $N$ is large enough
then $\mes\left(\TT\setminus\Omega_{N}\right)\le N\exp\left(-c_{1}N^{\sigma/2}\right)<\exp\left(-cN^{\sigma/2}\right)$.

Let $u\left(x\right)=\log\left|\f_{N}\left(x\right)\right|/N$ and
set $M=\left[N^{\sigma/8}/l_{0}\right]$. For each $x\in\Omega_{N}$
we have that $\left|\f_{N}\left(x+kl_{0}\omega\right)\right|>\exp\left(NL_{N}-100N^{\sigma}\right)/3$
for all but at most two $k$'s, $1\le k\le M$. We have
\begin{multline}
\left\langle u\right\rangle :=\int_{\TT}u\left(x\right)dx=\frac{1}{M}\sum_{k=1}^{M}\int_{\TT}u\left(x+kl_{0}\omega\right)dx\\
\ge\int_{\Omega_{N}}\left(\frac{M-2}{M}\left(L_{N}-100N^{\sigma-1}-\frac{\log3}{N}\right)+\frac{2}{M}\inf_{1\le k\le M}u\left(x+kl_{0}\omega\right)\right)dx\\
+\frac{1}{M}\sum_{k=1}^{M}\int_{\TT\setminus\Omega_{N}}u\left(x+kl_{0}\omega\right)dx.\label{eq:tu_first_estimate}
\end{multline}

Let $v\left(x\right)=\log\left|\g_{N}\left(x\right)\right|/N$. We
have that 
\[
S:=\sup_{z\in\mathcal{A}_{\rho_{0}''}}v\left(z\right)\le\sup_{z\in\mathcal{A}_{\rho_{0}''}}\frac{1}{N}\log\left\Vert \gh_{N}\left(z\right)\right\Vert <\infty.
\]
Let
\[
v\left(z\right)=\int_{\mathcal{A}_{\rho_{0}'}}\log\left|z-\zeta\right|d\mu\left(\zeta\right)+h\left(z\right)
\]
be the Riesz representation on $\cA_{\rho_{0}'}$. Applying \cite[Lemma 2.2]{MR2438997}
(see the proof of \lemref{M^a-Riesz_bounds}) we get that 
\begin{equation}
\mu\left(\mathcal{A}_{\rho_{0}}\right)+\left\Vert h\right\Vert _{L^{\infty}\left(\mathcal{A}_{\rho_{0}}\right)}\le C\left(2S-\sup_{\TT}v\right)\le C\left(2S-\left\langle v\right\rangle \right).\label{eq:f^a-Riesz-bounds}
\end{equation}
Note that $\left\langle v\right\rangle $ is finite by subharmonicity.
Since $\left\langle v\right\rangle =\left\langle u\right\rangle +D$,
it follows that $\left\langle u\right\rangle $ is also finite. Using
Cartan's estimate (see \cite[Lemma 2.2]{MR1847592}) we get that for
any small $\epsilon>0$ we have
\begin{equation}
\inf_{1\le k\le M}v\left(x+kl_{0}\omega\right)\ge-C\left(2S-\left\langle v\right\rangle \right)N^{\epsilon}\label{eq:inf_u}
\end{equation}
up to a set not exceeding $CM\exp\left(-N^{\epsilon}\right)$ in measure.
Since 
\begin{equation}
u\left(x\right)=v\left(x\right)-\frac{1}{2N}\left(S_{N}\left(x\right)+S_{N}\left(x+\omega\right)\right)\label{eq:u...v}
\end{equation}
we can use \eqref{inf_u} and \thmref{sh_ldt} to conclude that
\[
\inf_{1\le k\le M}u\left(x+kl_{0}\omega\right)>-C\left(2S-\left\langle u\right\rangle -D\right)N^{\epsilon}-D-N^{\epsilon}>\left(C\left\langle u\right\rangle -C'\right)N^{\epsilon}
\]
up to a set $\mathcal{B}_{N}$ not exceding $\exp\left(-cN^{\epsilon}\right)$
in measure. Therefore
\begin{multline*}
\int_{\Omega_{N}}\inf_{1\le k\le M}u\left(x+kl_{0}\omega\right)dx>\left(C\left\langle u\right\rangle -C'\right)N^{\epsilon}+\int_{\Omega_{N}\cap\mathcal{B}_{N}}\inf_{1\le k\le M}u\left(x+kl_{0}\omega\right)\\
>\left(C\left\langle u\right\rangle -C'\right)N^{\epsilon}-\sum_{k=1}^{M}\int_{\Omega_{N}\cap\mathcal{B}_{N}}\left|u\left(x+kl_{0}\omega\right)\right|dx.
\end{multline*}
Now \eqref{tu_first_estimate} becomes
\begin{multline*}
\left\langle u\right\rangle \ge\left(1-\frac{2}{M}\right)\left(L_{N}-100N^{\sigma-1}-\frac{\log3}{N}\right)+\frac{\left(C\left\langle u\right\rangle -C'\right)N^{\epsilon}}{M}\\
-\frac{2}{M}\sum_{k=1}^{M}\int_{\Omega_{N}^{c}\cup\mathcal{B}_{N}}\left|u\left(x+kl_{0}\omega\right)\right|.
\end{multline*}
Using \lemref{small_fl_is_small} (with $\sigma=3$) and reasoning
as in the proof of \lemref{Integral-lemma} we get that $\left\Vert u\right\Vert _{L^{2}\left(\TT\right)}\le CN^{3}$
and consequently
\[
\int_{\Omega_{N}^{c}\cup\mathcal{B}_{N}}\left|u\left(x+kl_{0}\omega\right)\right|dx\le\left(\mes\left\{ \Omega_{N}^{c}\cup\mathcal{B}_{N}\right\} \right)^{1/2}\left\Vert u\right\Vert _{L^{2}\left(\TT\right)}\le CN^{3}\exp\left(-cN^{\epsilon}\right).
\]
Now it is straightforward to reach the conclusion.\end{proof}
\begin{cor}
\label{cor:f^a-Riesz-bounds}Let 
\[
\frac{1}{n}\log\left\Vert \g_{n}\left(z\right)\right\Vert =\int_{\cA_{\rho_{0}'}}\log\left|z-\zeta\right|d\mu_{n}\left(\zeta\right)+h_{n}\left(z\right)
\]
be the Riesz representation on $\cA_{\rho_{0}'}$. There exists a
constant $C_{0}=C_{0}(\left\Vert a\right\Vert _{\infty},I_{a,E},$
$\left\Vert b\right\Vert _{*},\left|E\right|,\omega,\gamma,\rho_{0},\rho_{0}',\rho_{0}'')$
such that
\[
\mu_{n}\left(\cA_{\rho_{0}}\right)+\left\Vert h_{n}\right\Vert _{L^{\infty}\left(\cA_{\rho_{0}}\right)}\le C_{0}.
\]
\end{cor}
\begin{proof}
It suffices to obtain the bound for large $n$. The bound follows
from \eqref{f^a-Riesz-bounds} and the previous lemma.\end{proof}
\begin{lem}
\label{lem:ldt_entries_weak}There exist constants $\sigma_{0}>0$,
$c_{0}=c_{0}\left(I_{a,E},\left\Vert b\right\Vert _{*},\left|E\right|,\omega,\gamma\right)$,
and $C_{0}=C_{0}\left(I_{a,E},\left\Vert b\right\Vert _{*},\left|E\right|,\omega,\gamma\right)$
such that for every integer $n$ and any $\delta>0$ we have
\[
\mes\left\{ x\in\TT:\,\left|\log\left|\g_{n}\left(x\right)\right|-\left\langle \log\left|\g_{n}\right|\right\rangle \right|>n\delta\right\} \le C_{0}\exp\left(-c_{0}\delta n^{\sigma_{0}}\right).
\]
The same estimate with possibly different $c_{0}$ and $C_{0}$ holds
for $\f_{n}$.\end{lem}
\begin{proof}
It is enough to establish the estimate for $n$ large enough. Let
$u\left(x\right)=\log\left|\f_{n}\left(x\right)\right|/n$ and $v\left(x\right)=\log\left|\g_{n}\left(x\right)\right|/n$.
By the previous lemma (recall that $\left\langle v\right\rangle =\left\langle u\right\rangle +D$)
and \propref{M^a-upper-bound} we have that there exists a small $\kappa>0$
such that
\[
\begin{cases}
\left\langle v\right\rangle \ge L_{n}^{a}-n^{-\kappa}\\
\sup_{\TT}v\le L_{n}^{a}+n^{-\kappa}.
\end{cases}
\]
This implies that
\[
\left\Vert v-\left\langle v\right\rangle \right\Vert _{L^{1}\left(\TT\right)}\le Cn^{-\kappa}
\]
and hence by \cite[Lemma 2.3]{MR1843776} we have
\[
\left\Vert v\right\Vert _{BMO\left(\TT\right)}=\left\Vert v-\left\langle v\right\rangle \right\Vert _{BMO\left(\TT\right)}\le C\left\Vert v-\left\langle v\right\rangle \right\Vert _{L^{1}\left(\TT\right)}^{1/2}\le Cn^{-\kappa/2}.
\]
As in the proof of \cite[Proposition 2.11]{MR2438997} we note that
in order to get the conclusion of \cite[Lemma 2.3]{MR1843776} we
just need the bounds on the Riesz representation of $v$. By the John-Nirenberg
inequality we get
\[
\mes\left\{ x\in\TT:\,\left|v\left(x\right)-\left\langle v\right\rangle \right|>\delta\right\} \le C\exp\left(-c\delta n^{\kappa/2}\right).
\]
Using \eqref{u...v} we have
\begin{multline*}
\mes\left\{ x\in\TT:\,\left|u\left(x\right)-\left\langle u\right\rangle \right|>\delta\right\} \le\mes\left\{ x\in\TT:\,\left|v\left(x\right)-\left\langle v\right\rangle \right|>\frac{\delta}{2}\right\} \\
+\mes\left\{ x\in\TT:\,\left|\frac{1}{2n}\left(S_{n}\left(x\right)+S_{n}\left(x+\omega\right)\right)-D\right|>\frac{\delta}{2}\right\} \\
\le C\exp\left(-c\delta n^{\kappa/2}/2\right)+2\exp\left(-c'\delta n/2+r_{n}\right)\le C'\exp\left(-c''\delta n^{\kappa/2}/2\right).
\end{multline*}
This concludes the proof.
\end{proof}
Next we will use the Avalanche Principle to refine the previous estimate.
\begin{prop}
\label{prop:ldt_entries} There exist constants $c_{0}=c_{0}\left(\left\Vert a\right\Vert _{\infty},I_{a,E},\left\Vert b\right\Vert _{*},\left|E\right|,\omega,\gamma\right)$,
$C_{0}=C_{0}\left(\omega\right)>\alpha+2$, and $C_{1}=C_{1}\left(\left\Vert a\right\Vert _{\infty},I_{a,E},\left\Vert b\right\Vert _{*},\left|E\right|,\omega,\gamma\right)$
such that for every integer $n>1$ and any $\delta>0$ we have
\[
\mes\left\{ x\in\TT:\,\left|\log\left|\g_{n}\left(x\right)\right|-\left\langle \log\left|\g_{n}\right|\right\rangle \right|>n\delta\right\} \le C_{1}\exp\left(-c_{0}\delta n\left(\log n\right)^{-C_{0}}\right).
\]
\end{prop}
\begin{proof}
It is enough to establish the estimate for $n$ large enough. We have
that
\[
\left[\begin{array}{cc}
\f_{n}\left(x\right) & 0\\
0 & 0
\end{array}\right]=\left[\begin{array}{cc}
-1 & 0\\
0 & 0
\end{array}\right]\fg_{n}\left(x\right)\left[\begin{array}{cc}
-1 & 0\\
0 & 0
\end{array}\right]=:\cM_{n}^{u}\left(x\right).
\]
We define $\cM_{n}^{a}$ analogously. We obviously have that $\left|f_{n}^{a}\left(x\right)\right|=\left\Vert \cM_{n}^{a}\left(x\right)\right\Vert $. 

Let $l=\left[\left(\log n\right)^{2/\sigma_{0}}\right]$ with $\sigma_{0}$
as in \lemref{ldt_entries_weak}. Let $n=l+\left(m-2\right)l+l'$
with $2l\le l'\le3l$. We want to apply the Avalanche Principle to
$\cM_{n}^{u}\left(x\right)=\prod_{j=m}^{1}A_{j}^{u}\left(x\right)$
where $A_{j}^{u}\left(x\right)=\fg_{l}\left(x+\left(j-1\right)l\omega\right)$,
$j=2,\ldots,m-1$,
\[
A_{1}^{u}\left(x\right)=\fg_{l}\left(x\right)\left[\begin{array}{cc}
-1 & 0\\
0 & 0
\end{array}\right]=\left[\begin{array}{cc}
\f_{l}\left(x\right) & 0\\
\star & 0
\end{array}\right],
\]
and
\[
A_{m}^{u}\left(x\right)=\left[\begin{array}{cc}
-1 & 0\\
0 & 0
\end{array}\right]\fg_{l'}\left(x\right)=\left[\begin{array}{cc}
\f_{l'}\left(x\right) & \star\\
0 & 0
\end{array}\right].
\]
We define the matrices $A_{j}^{a}$ analogously. We clearly have that
\[
\log\left|\f_{l}\left(x\right)\right|\le\log\left\Vert A_{1}^{u}\left(x\right)\right\Vert \le\log\left\Vert \fg_{l}\left(x\right)\right\Vert ,
\]
and an analogous estimate for $\log\left\Vert A_{m}^{u}\right\Vert $.
Now it follows from \corref{ldt_M^u}, \lemref{ldt_entries_weak},
and \lemref{avg_entries->lower_bound} that the hypotheses of \lemref{AppliedAP}
are satisfied and hence
\[
\log\left\Vert \cM_{n}^{a}\left(x\right)\right\Vert +\sum_{j=2}^{m-1}\log\left\Vert A_{j}^{a}\left(x\right)\right\Vert -\sum_{j=1}^{m-1}\log\left\Vert A_{j+1}^{a}\left(x\right)A_{j}^{a}\left(x\right)\right\Vert =O\left(\frac{1}{l}\right)
\]
up to a set of measure less than $3n\exp\left(-cl^{\sigma_{0}}\right)<\exp\left(-c'\left(\log n\right)^{2}\right)$.
Note that, as before, we checked the conditions of the Avalanche Principle
for $\cM_{n}^{u}$, but we wrote the conclusion for $\cM_{n}^{a}$.
By letting 
\[
u_{0}\left(x\right)=\log\left\Vert A_{m}^{a}\left(x\right)A_{m-1}^{a}\left(x\right)\right\Vert +\log\left\Vert A_{2}^{a}\left(x\right)A_{1}^{a}\left(x\right)\right\Vert 
\]
 we rewrite the previous relation as
\begin{multline*}
\log\left\Vert \cM_{n}^{a}\left(x\right)\right\Vert +\sum_{j=2}^{m-1}\log\left\Vert \gh_{l}\left(x+\left(j-1\right)l\omega\right)\right\Vert \\
-\sum_{j=2}^{m-2}\log\left\Vert \gh_{2l}\left(x+\left(j-1\right)l\omega\right)\right\Vert -u_{0}\left(x\right)=O\left(\frac{1}{l}\right).
\end{multline*}
Note that
\begin{multline}
\log\left\Vert \g_{l+l'}\left(x+\left(m-2\right)l\omega\right)\right\Vert +\log\left\Vert \g_{2l}\left(x\right)\right\Vert \le u_{0}\left(x\right)\\
\le\log\left\Vert \gh_{l+l'}\left(x+\left(m-2\right)l\omega\right)\right\Vert +\log\left\Vert \gh_{2l}\left(x\right)\right\Vert .\label{eq:u-bounds}
\end{multline}
We apply the Avalanche Principle $l-1$ more times. At each step we
decrease the length of $A_{m}$ by one and increase the length of
$A_{1}$ by one. Adding the resulting estimates and dividing by $l$
yields
\begin{multline}
\log\left\Vert \cM_{n}^{a}\left(x\right)\right\Vert +\sum_{j=l}^{\left(m-1\right)l-1}\frac{1}{l}\log\left\Vert M_{l}^{a}\left(x+j\omega\right)\right\Vert -\sum_{j=l}^{\left(m-2\right)l-1}\frac{1}{l}\log\left\Vert M_{2l}^{a}\left(x+j\omega\right)\right\Vert \\
-\sum_{k=0}^{l-1}\frac{1}{l}u_{k}\left(x\right)=O\left(\frac{1}{l}\right)\label{eq:M_impure_AP}
\end{multline}
up to a set of measure less than $l\exp\left(-c\left(\log n\right)^{2}\right)<\exp\left(-c'\left(\log n\right)^{2}\right)$.
The functions $u_{k}$, $k=1,\ldots,l-1$ are defined analogously
to $u_{0}$ and satisfy estimates analogous to \eqref{u-bounds}.
Based on these estimates it is straightforward to conclude (see \lemref{M^a-Riesz_bounds}
and \eqref{f^a-Riesz-bounds}) that there is an uniform bound for
the Riesz representations of $u_{k}/l$, $k=1,\ldots,l-1$. Hence
we can use \thmref{sh_ldt} to get
\[
\sum_{k=0}^{l-1}\frac{1}{l}u_{k}\left(x\right)-\sum_{k=0}^{l-1}\frac{1}{l}\left\langle u_{k}\right\rangle =O\left(l\left(\log n\right)^{2}\right)=O\left(\left(\log n\right)^{2+2/\sigma_{0}}\right)
\]
up to a set of measure less than $l\exp\left(-c\left(\log n\right)^{2}\right)<\exp\left(-c'\left(\log n\right)^{2}\right)$.
On the other hand, using \thmref{ldt} we have
\begin{multline*}
\sum_{j=l}^{\left(m-1\right)l-1}\frac{1}{l}\log\left\Vert M_{l}^{a}\left(x+j\omega\right)\right\Vert -\sum_{j=l}^{\left(m-2\right)l-1}\frac{1}{l}\log\left\Vert M_{2l}^{a}\left(x+j\omega\right)\right\Vert \\
=\left(m-2\right)lL_{l}^{a}-\left(m-3\right)lL_{2l}^{a}+O\left(\left(\log n\right)^{p}\right)
\end{multline*}
up to a set of measure less than $\exp\left(-c\left(\log n\right)^{p}\right)$.
We can now conclude from \eqref{M_impure_AP} that
\[
\log\left|f_{n}^{a}\left(x\right)\right|+\left(m-2\right)lL_{l}^{a}-\left(m-3\right)lL_{2l}^{a}-\sum_{k=0}^{l-1}\frac{1}{l}\left\langle u_{k}\right\rangle =O\left(\left(\log n\right)^{C_{2}}\right)
\]
up to a set of measure less than $\exp\left(-c\left(\log n\right)^{2}\right)$,
where $C_{2}=\max\left\{ p,2+2/\sigma_{0}\right\} $. Integrating
the above relation and then subtracting it, yields
\begin{equation}
\left|\log\left|f_{n}^{a}\left(x\right)\right|-\left\langle \log\left|f_{n}^{a}\right|\right\rangle \right|\le C\left(\log n\right)^{C_{2}}\label{eq:f^a-avg-bound}
\end{equation}
up to a set of measure less than $\exp\left(-c\left(\log n\right)^{2}\right)$.
Note that the exceptional set was handled by using the fact that $\left\Vert \log\left|f_{n}^{a}\right|\right\Vert _{L^{2}\left(\TT\right)}\le Cn$.
This follows from 
\begin{equation}
\left\Vert \log\left|f_{n}^{a}\right|-\left\langle \log\left|f_{n}^{a}\right|\right\rangle \right\Vert _{L^{2}\left(\TT\right)}\le Cn\label{eq:f^a-avg-L^2}
\end{equation}
 and $\left|\left\langle \log\left|f_{n}^{a}\right|\right\rangle \right|\le Cn$.
The first estimate is an imediate consequence of \lemref{ldt_entries_weak}
and \lemref{Integral-lemma}. The second estimate can be deduced from
\lemref{avg_entries->lower_bound}.

Let $\cB$ be the exceptional set for \eqref{f^a-avg-bound}. Let
\[
\log\left|f_{n}^{a}\right|-\left\langle \log\left|f_{n}^{a}\right|\right\rangle =u_{0}+u_{1}
\]
 where $u_{0}=0$ on $\cB$ and $u_{1}=0$ on $\TT\setminus\cB$.
By \eqref{f^a-avg-bound} and \eqref{f^a-avg-L^2} we have that $\left\Vert u_{0}-\left\langle u_{0}\right\rangle \right\Vert _{L^{\infty}\left(\TT\right)}\le C\left(\log n\right)^{C_{2}}$
and 
\[
\left\Vert u_{1}\right\Vert _{L^{2}\left(\TT\right)}\le Cn\sqrt{\mes\left(\cB\right)}\le\exp\left(-c\left(\log n\right)^{2}\right).
\]
Applying \cite[Lemma 2.3]{MR1843776} we have
\[
\left\Vert \log\left|f_{n}^{a}\right|\right\Vert _{BMO\left(\TT\right)}\le C\left(\left(\log n\right)^{C_{2}+2}+\sqrt{n\exp\left(-c\left(\log n\right)^{2}\right)}\right)\le C'\left(\log n\right)^{C_{0}}.
\]
The conclusion follows from the John-Nirenberg inequality.\end{proof}
\begin{lem}
\label{lem:<f>-nL}There exists a constant $C_{0}=C_{0}\left(\left\Vert a\right\Vert _{\infty},I_{a,E},\left\Vert b\right\Vert _{*},\left|E\right|,\omega,\gamma\right)$
such that
\[
\left|\left\langle \log\left|f_{n}^{a}\right|\right\rangle -nL_{n}^{a}\right|\le C_{0}
\]
for all integers.\end{lem}
\begin{proof}
Subtracting the Avalanche Principle expansions for $\cM_{n}^{a}$
and $M_{n}^{a}$ at scale $l\thickapprox\left(\log n\right)^{A}$
and then integrating, yields
\[
\left|\left\langle \log\left|f_{n}^{a}\right|\right\rangle -nL_{n}^{a}\right|\le CR\left(4\left(\log n\right)^{A}\right)+O\left(\frac{1}{l}\right)
\]
where 
\[
R\left(n\right)=\sup_{n/2\le m\le n}\left|\left\langle \log\left|f_{m}^{a}\right|\right\rangle -mL_{m}^{a}\right|.
\]
Iterating this estimate yields the desired conclusion (cf. \cite[Lemma 3.5]{MR2438997}).
\end{proof}
We now prepare to prove the estimate on the number of eigenvalues.
Fix $E_{0}\in\mathbb{R}$ such that $L\left(E_{0}\right)\ge\gamma>0$.
As a consequence of \corref{L(E)-L(E0)} and \lemref{Ln_L} it follows
that there exists a disk $\cD$ around $E_{0}$ such that $L\left(E\right)\ge\gamma/2$
on $I$. In what follows we also fix $\cD$. Note that the existence
of the disk $\cD$ would follow from the continuity of the Lyapunov
exponent, which is known from \cite{MR2825743}. However, we also
need the information on the modulus of continuity provided by \corref{L(E)-L(E0)}.
This information follows from the Hölder continuity of the Lyapunov
exponent proved in \cite{2011arXiv1108.3747T}, but we use \corref{L(E)-L(E0)}
in order to keep the paper self-contained. The following deviations
estimate in $E$ will be needed in the proof of the estimate. 
\begin{lem}
\label{lem:ldt_in_E} Let $C_{0}$ be as in \propref{ldt_entries}.
There exist constants $c_{0}=c_{0}\left(\left\Vert a\right\Vert _{\infty},I_{a,E},\left\Vert b\right\Vert _{*},\left|E\right|,\omega,\gamma\right)$
and $C_{1}=C_{1}\left(\left\Vert a\right\Vert _{\infty},I_{a,E},\left\Vert b\right\Vert _{*},\left|E\right|,\omega,\gamma\right)$
such that for every integer $n>1$ and every $\delta\ge\delta_{0}$
there exists a set $\cB_{n,\delta}\subset\TT$ with $\mes\cB_{n,\delta}<C_{1}\exp\left(-c_{0}\delta\left(\log n\right)^{-C_{0}}\right)$,
such that for each $x\in\TT\setminus\cB_{n,\delta}$ there exists
$\cE_{n,\delta,x}\subset\cD$, with $\mes\cE_{n,\delta,x}<C_{1}\exp\left(-c_{0}\delta\left(\log n\right)^{-C_{0}}\right)$,
such that
\begin{equation}
\left|\log\left|f_{n}^{a}\left(x,E\right)\right|-nL_{n}^{a}\left(E\right)\right|\le\delta,\label{eq:ldt_in_E}
\end{equation}
for any $E\in\cD\setminus\cE_{n,\delta,x}$.\end{lem}
\begin{proof}
From \propref{ldt_entries} and \lemref{<f>-nL} it follows that \eqref{ldt_in_E}
holds for for $\delta\ge\delta_{0}$, and $\left(x,E\right)\in\TT\times\cD$
except for a set of measure $C\exp\left(-c\delta/\left(\log n\right)^{C_{0}}\right)$.
The conclusion follows by Fubini's Theorem and Chebyshev's inequality.\end{proof}
\begin{thm}
\label{thm:Number-of-eigenvalues} Let $C_{0}=C_{0}\left(\omega\right)$
be as in \propref{ldt_entries}. There exist constants $C_{1}=C_{1}\left(\left\Vert a\right\Vert _{\infty},\left\Vert b\right\Vert _{*},\left|E_{0}\right|,\omega,\gamma\right)$,
$C_{2}=C_{2}\left(\left\Vert a\right\Vert _{\infty},I_{a,\cD},\left\Vert b\right\Vert _{*},\left|E_{0}\right|,\omega,\gamma\right)$,
and $n_{0}=n_{0}\left(\left\Vert a\right\Vert _{\infty},I_{a,\cD},\left\Vert b\right\Vert _{*},\left|E_{0}\right|,\omega,\gamma\right)$
such that for any $x_{0}\in\TT$, $E_{0}\in\mathbb{R}$, and $n\ge n_{0}$
one has 
\[
\#\left\{ E\in\mathbb{R}:\, f_{n}^{a}\left(x_{0},E\right)=0,\,\left|E-E_{0}\right|<n^{-C_{1}}\right\} \le C_{2}\left(\log n\right)^{2C_{0}}
\]
and
\[
\#\left\{ z\in\mb C:\, f_{n}^{a}\left(z,E_{0}\right)=0,\,\left|z-x_{0}\right|<n^{-1}\right\} \le C_{2}\left(\log n\right)^{2C_{0}}.
\]
\end{thm}
\begin{proof}
From \eqref{ldt_in_E} it follows that there exist $x_{1},\, E_{1}$
such that 
\[
\left|x_{1}-x_{0}\right|\le C\exp\left(-c\left(\log n\right)^{C_{0}}\right),
\]
\[
\left|E_{1}-E_{0}\right|\le C\exp\left(-c\left(\log n\right)^{C_{0}}\right),
\]
and
\begin{equation}
\log\left|f_{n}^{a}\left(x_{1},E_{1}\right)\right|\ge nL_{n}\left(E_{1}\right)-\left(\log n\right)^{2C_{0}}.\label{eq:f^a-lb}
\end{equation}
 Let $R=n^{-2C_{3}}$, where $C_{3}$ is the constant $C_{1}$ from
\corref{L(E)-L(E0)}, and let $\nu_{x,E}\left(r\right)=\#\left\{ E:\, f_{n}^{a}\left(x,E'\right)=0,\,\left|E'-E\right|\le r\right\} $.
Using Jensen's formula we have that
\begin{multline}
\nu_{x_{1},E_{1}}\left(3R\right)\le C\int_{0}^{4R}\frac{\nu_{x_{1},E_{1}}\left(t\right)}{t}dt\\
=\frac{1}{2\pi}\int_{0}^{2\pi}\log\left|f_{n}^{a}\left(x_{1},E_{1}+4Re^{i\theta}\right)\right|d\theta-\log\left|f_{n}^{a}\left(x_{1},E_{1}\right)\right|\label{eq:Jensen-at-E1}
\end{multline}
By \propref{M^a-upper-bound} we have 
\[
\log\left|f_{n}^{a}\left(x_{1},E\right)\right|\le nL_{n}^{a}\left(E\right)+C\left(\log n\right)^{p}
\]
for $E\in\cD$. Using this, together with \eqref{f^a-lb} and \eqref{Jensen-at-E1}
yields
\[
\nu_{x_{1},E_{1}}\left(3R\right)\le C\left(\sup_{\left|E-E_{1}\right|=4R}\left(n\left(L_{n}^{a}\left(E\right)-L_{n}^{a}\left(E_{1}\right)\right)\right)+\left(\log n\right)^{2C_{0}}\right).
\]
For $E$ such that $\left|E-E_{0}\right|\le R$ we have that $\left|E-E_{1}\right|\le n^{-C_{3}}$
and hence by \corref{L(E)-L(E0)} we have that $\left|n\left(L_{n}^{a}\left(E\right)-L_{n}^{a}\left(E_{1}\right)\right)\right|\le n^{-C}$.
We can now conclude that 
\begin{equation}
\nu_{x_{1},E_{0}}\left(2R\right)\le\nu_{x_{1},E_{1}}\left(3R\right)\le C\left(\log n\right)^{2C_{0}}.\label{eq:z(R)-bound}
\end{equation}
Using the Mean Value Theorem we can conclude that
\[
\left\Vert H^{\left(n\right)}\left(x_{0}\right)-H^{\left(n\right)}\left(x_{1}\right)\right\Vert \le C\left|x_{0}-x_{1}\right|\le C\exp\left(-c\left(\log n\right)^{C_{0}}\right).
\]
Let $E_{j}^{\left(n\right)}\left(x\right)$, $j=1,\ldots,n$ be the
eigenvalues of $H^{\left(n\right)}\left(x\right)$ ordered increasingly.
Since $H^{\left(n\right)}\left(x_{0}\right)$ and $H^{\left(n\right)}\left(x_{1}\right)$
are Hermitian it follows that 
\[
\left|E_{j}^{\left(n\right)}\left(x_{0}\right)-E_{j}^{\left(n\right)}\left(x_{1}\right)\right|\le C\exp\left(-c\left(\log n\right)^{C_{0}}\right).
\]
This implies that $\nu_{x_{0},E_{0}}\left(R\right)\le\nu_{x_{1},E_{0}}\left(2R\right)$
and now the first estimate follows from \eqref{z(R)-bound}.

The second estimate follows in a similar way. From \propref{ldt_entries}
it follows that there exists $x_{1}$ such that $\left|x_{1}-x_{0}\right|\le C\exp\left(-c\left(\log n\right)^{C_{0}}\right)$
and 
\begin{equation}
\log\left|f_{n}^{a}\left(x_{1},E_{0}\right)\right|\ge nL_{n}\left(E_{0}\right)-\left(\log n\right)^{2C_{0}}.\label{eq:f^a(E0)-lb}
\end{equation}
Let $\nu_{x}\left(r\right)=\#\left\{ z\in\mb C:\, f_{n}^{a}\left(z,E_{0}\right)=0,\,\left|z-x\right|<r\right\} $.
Using Jensen's formula, \eqref{f^a(E0)-lb} and \propref{M^a-upper-bound},
as before, yields
\begin{multline*}
\nu_{x_{0}}\left(n^{-1}\right)\le\nu_{x_{1}}\left(2n^{-1}\right)\\
\le C\left(\sup_{r\in\left(1-3n^{-1},1+3n^{-1}\right)}\left(n\left(L_{n}^{a}\left(r,E_{0}\right)-L_{n}^{a}\left(1,E_{0}\right)\right)\right)+\left(\log n\right)^{2C_{0}}\right)\\
\le C'\left(\log n\right)^{2C_{0}}.
\end{multline*}
For the last inequality we used \corref{L(r1)-L(r2)}. This concludes
the proof. 
\end{proof}
\bibliographystyle{halpha}
\bibliography{schrodinger}

\vspace{2cm}

\begin{flushleft}
\textbf{I. Binder}: Dept. of Mathematics, University of Toronto, Toronto,
ON, M5S 2E4, Canada; \texttt{ilia@math.utoronto.ca}
\par\end{flushleft}

\medskip{}

\begin{flushleft}
\textbf{M. Voda}: Dept. of Mathematics, University of Toronto, Toronto,
ON, M5S 2E4, Canada; \texttt{mvoda@math.utoronto.ca}
\par\end{flushleft}
\end{document}